\documentclass[format=acmsmall, review=false]{acmart}
\usepackage{booktabs} 
\usepackage[title]{appendix}
\usepackage{thm-restate}
\usepackage{pict2e}
\usepackage[ruled,noend]{algorithm2e} 

\SetAlFnt{\small}
\SetAlCapFnt{\small}
\SetAlCapNameFnt{\small}
\SetAlCapHSkip{0pt}
\IncMargin{-\parindent}
\usepackage{subcaption}
\usepackage{MnSymbol}
\setcitestyle{authoryear}

    \usepackage{xcolor}
    \usepackage{derivative}
    \usepackage{graphicx} 
    \usepackage{amsfonts,amsmath,amsthm}

\DeclareRobustCommand{\legendsquare}[1]{%
  \textcolor{#1}{\rule{1.5ex}{1.5ex}}%
}
\definecolor{myred}{RGB}{158,11,55}
\definecolor{myblue}{RGB}{61,156,220}
\colorlet{mygrey}{gray}

    \newtheorem{theorem}{Theorem}
    
    \newtheorem{proposition}{Proposition}
    \newtheorem{claim}{Claim}
    \newtheorem{observation}{Observation}

\theoremstyle{definition}
    
\newtheorem{example}{Example}
\AtBeginEnvironment{example}{%
  \pushQED{\qed}%
}
\AtEndEnvironment{example}{\popQED\endexample}

\newenvironment{claimproof}[1]{\par\noindent\textit{Proof.}\space#1}{\hfill$\lhd$}

\newcommand{\myparagraph}[1]{
\smallskip
\noindent 
{\bf #1}}

\newcommand{\EE}{E^{\mathrm{net}}}
\def\PriPropClear{\textsc{Priority-Proportional Clearing}}
\def\FindComp{\textsc{Optimal Compression for Proportional Clearing}}
\def\FindCompPairs{\textsc{Optimal Bilateral Compression for Proportional Clearing}}
\def\FindCompCycle{\textsc{Optimal Compression Cycle for Proportional Clearing}}
\def\FindCompBank{\textsc{$b$-optimal Compression for Proportional Clearing}}

\newcommand{\ildi}[1]{\textcolor{blue}{#1  [Ildi]}}

\def\nondeft{\textup{nondef}}
\def\deft{\textup{def}}

    \newcommand{\alp}{\widetilde{\alpha}}
    \newcommand{\bet}{\widetilde{\beta}}
    
    \newcommand{\delz}{\triangle z}
    \newcommand{\delp}{\triangle p}
    \def\deg{d}
    
    \newcommand{\gnote}{\textcolor{green!50!black}}
    \newcommand{\maxsat}{{\sc Max-2-SAT}}
    \newcommand{\partition}{{\sc Partition}}
    \def\dft{\mathcal{D}}
    \newcommand{\twosat}{{\sc 2-sat}}

    \def\P{\mathcal{P}}
    \def\NP{\mathsf{NP}}
    
\newcommand{\linkproof}[1]{%
\hyperref[#1]{$\star$}%
}

\def\mypropsymbol{\spadesuit}

\newcommand{\opentriangle}{%
  \raisebox{0.2pt}{\makebox[0.77778em]{%
    \setlength{\unitlength}{0.6em}%
    \linethickness{0.4pt}\roundjoin
    \begin{picture}(1,1)
    \polygon(0,0)(1,0)(1,1)
    \end{picture}%
  }}%
}

    \title{Optimal Portfolio Compression for Priority-Proportional Clearing with Defaulting Costs}
\author{Gergely Csáji}
\affiliation{%
  \institution{ELTE Centre for Economic and Regional Studies}
  \city{Budapest}
  \country{Hungary}
}
\affiliation{%
  \institution{Budapest University of Technology and Economics}
  \city{Budapest}
  \country{Hungary}
}
\email{csaji.gergely@krtk.elte.hu}
 \author{Rareş-Ioan Mateiu}
\affiliation{%
  \institution{University of Bucharest}
  \city{Bucharest}
  \country{Romania}
  }
\email{raresmateiu14@gmail.com}
\author{Alexandru Popa}
\affiliation{%
  \institution{University of Bucharest}
  \city{Bucharest}
  \country{Romania}
  }
\email{alexandru.popa@fmi.unibuc.ro}

\author{Ildikó Schlotter}
\affiliation{%
  \institution{ELTE Centre for Economic and Regional Studies}
  \city{Budapest}
  \country{Hungary}
}
\affiliation{%
  \institution{Budapest University of Technology and Economics}
  \city{Budapest}
  \country{Hungary}
}
\email{schlotter.ildiko@krtk.elte.hu}

    \begin{abstract}
We study financial networks where banks are connected through bilateral liabilities and may default when resources are insufficient to meet obligations. We consider both the standard proportional clearing model and a priority-proportional clearing model in which banks repay creditors according to exogenously given priority classes. 
In such markets, portfolio compression is a process where  several banks come to a netting arrangement which reduces liabilities without changing any bank's net exposure, essentially removing cycles of debt.
Our goal is to understand whether portfolio compression schemes can be designed to improve clearing outcomes for a large fraction of banks.

We provide a computational characterization of the benefits and limitations of compression. On the positive side, we give a polynomial-time algorithm to compute a maximal clearing outcome under priority-proportional clearing, and we show that it is possible to decide in polynomial time whether there exists a compression that limits defaults to at most one bank.
On the negative side, we show that several natural optimization and decision problems are computationally intractable: deciding whether some compression can reduce the number of defaulting banks below a given threshold, or whether a specific bank can be saved from defaulting, is $\NP$-hard even in restricted settings and under proportional clearing.

We further present a mixed integer linear programming (MILP) formulation that computes a compression maximizing the number of non-defaulting banks, providing a practical approach to this hard problem. Using our MILP formulation, we perform simulations on both synthetic and real-world datasets to analyze the effects of portfolio compression.

    \end{abstract}

\begin{document}

    \maketitle

\begin{titlepage}

\vspace{1cm}
\setcounter{tocdepth}{1} 
\tableofcontents

\end{titlepage}
    
    \section{Introduction}
    
Following the 2008 global financial crisis, regulatory focus intensified on mitigating systemic risk in the complex, opaque Over-the-Counter (OTC) derivatives markets~\cite{acharya2012doddfrank}. A key objective has been the reduction of the vast notional outstanding exposure, which creates significant operational, credit, and liquidity risk~\cite{duffie2017role,helleiner2021governing}. Multilateral netting and portfolio compression are the primary post-trade mechanisms designed to achieve this reduction without altering the market participants' net financial obligations.

Multilateral netting is a generalized process where offsetting cash flows or contractual obligations among three or more counterparties are consolidated into a single, smaller net payment stream~\cite{bis1990report,garratt2020centralized}. This technique is not limited to financial markets but has also been adapted for use in public administration to manage inter-company debts, such as the novel algorithm applied within the Romanian Ministry of Economy~\cite{gavrila2021novel}.

 Portfolio compression, which is widely utilized by specialized service providers, is a highly specific application of multilateral netting for derivatives~\cite{derrico2021compressing,o2017optimising}. In a compression cycle, multiple existing contracts are simultaneously terminated and subsequently replaced by a significantly smaller set of new contracts (often a single contract) that strictly preserves the original net market value for every participant. This method efficiently reduces gross notional outstanding and associated risks, maintaining net exposure while drastically cutting gross exposure~\cite{veraart2022does}.

\subsection{Related Work}
 We briefly review the literature surrounding this topic, which we group into three main categories.

\myparagraph{Mathematical and computational optimization of compression.}
The core problem is finding a cycle or a set of contracts to eliminate that yields the maximum reduction in gross notional or regulatory capital consumption. 

\citet{eisenberg2001systemic} introduce the mathematical framework for clearing payments in a network of defaulting firms (i.e., in particular the notion of \emph{clearing vector}). They provided an algorithm to find a maximum proportional clearing vector, which requires banks to pay liabilities proportional to the amounts owed, when they default. Later on,  
\citet{o2017optimising}  addressed the computational challenge of reducing systemic risk in Over-the-Counter (OTC) derivatives markets by treating portfolio compression as a formal network optimization problem. Focusing on fungible derivatives (standardized contracts where individual trades can be netted), the author proposes and evaluates several algorithms designed to minimize gross counterparty exposure while strictly preserving the net market risk of every participant.

\citet{derrico2021compressing} analyze the compression of OTC markets by introducing mathematical models that quantify how different compression mechanisms reduce the size and change the structure of the debt network. They provide computational insights into the efficacy of various approaches in eliminating cycles of exposures.

A more formalized view of the optimization challenge is presented by 
\citet{amini2023optimal}, who study the problem of Optimal Network Compression. They focus on the computational complexity of finding the best set of contracts to compress within a financial network. The authors demonstrate that finding a globally optimal compression that minimizes a complex metric (such as a measure of systemic risk) is generally an $\NP$-hard problem, underscoring the practical limitations of optimizing for system-wide stability. We improve on their results showing that optimal compressions are $\NP$-hard to find even for simple objectives such as minimizing the banks that default, or saving a given bank from bankruptcy. 


A recent framework for priority-proportional clearing was introduced by \citet{KanellopoulosKZ24}. While the authors claim a polynomial-time algorithm for computing maximal clearing payments, their procedure (Algorithm 2, Step 5) relies on the computation of limit points for an iterative series. Crucially, the authors do not provide a constructive mechanism for finding these limits in polynomial time, nor a bound on the iterations required for convergence. Given that these limits are not simple linear series, their computation remains the primary technical hurdle. Our work provides the first truly constructive polynomial-time
algorithm by demonstrating that these limit points can be explicitly found using linear programming, thereby closing this computational gap. 
Furthermore, we must highlight a factual mistake in their proof correctness. They explicitly assume that the initial endowments ``$e_i$ may also be negative and corresponds to financial obligations towards entities outside the system in consideration''; however, this is in contradiction with their claim that Equation (5) in their paper (within Lemma 2) is guaranteed to have a nonnegative solution. Indeed, assume that the market contains a single bank~$b$, with $e_b=-1$ and no liabilities ($L_b=0$). Then, the set of initially defaulting banks in their construction is $D_0 = \{ 1\}$, and so the unique solution to their Equation (5) is where $x_i = \alpha e_i <0$.

\myparagraph{The structural analysis of centralized clearing and financial networks.}
Portfolio compression services often operate in conjunction with Central Counterparties (CCPs) or other centralized clearing mechanisms, which form the backbone of the organized financial network.

\citet{csoka2022centralized} investigate centralized clearing mechanisms in financial networks using a programming approach. Their work connects the institutional structure of a centralized clearing house, which acts as a hub for multilateral netting, with the mathematical principles of efficiency and debt resolution, suggesting that regulatory goals can be achieved through properly formulated optimization programs.

The effectiveness of centralizing risk through netting is critically examined by 
\citet{garratt2020centralized}. In their study on centralized netting in financial networks, they show that while centralized netting generally reduces the variance of net exposures (making risk more predictable), it does not always reduce the expected value of net exposures, particularly for certain network topologies or when the number of asset classes is high. This introduces a subtle trade-off, where centralized netting might not always be individually beneficial to all participants, potentially explaining a historical reluctance to adopt central clearing prior to regulation.

\myparagraph{The relationship between compression and systemic risk.}
%
\citet{veraart2022does} investigates when portfolio compression reduces systemic risk, establishing that compression is not a universally beneficial tool. If fragile  that are likely to default participate in compression, the resulting change in the payment network structure can, under certain conditions, increase the systemic risk of contagion. The paper derives structural conditions related to the resilience of participating nodes, recovery rates, and debt repayment capacity that determine if a compression is harmful or beneficial.

\citet{schuldenzucker2021portfolio} study the impact of compression on contagion in financial networks and highlight the role of incentives. They show that while compression reduces gross exposures (a proxy for counterparty credit risk), it can restructure the interdependencies in ways that exacerbate the overall risk profile of the network. Furthermore, they address the challenge of ensuring participants voluntarily engage in compression when the benefits (especially the reduction of systemic risk) are often externalized.

Thus, \citet{veraart2022does} and \citet{schuldenzucker2021portfolio} show that regulatory intervention must look beyond simple notional reduction and consider the quality of the counterparties involved and the resulting changes in the network's vulnerability to cascade failures.

\myparagraph{Other related work.}
As mentioned before, the theoretical foundation for analyzing default cascades was initiated by 
\citet{eisenberg2001systemic}, who established the existence of a unique clearing vector in interbank networks.  
\citet{rogers2013failure} extended the classical framework of~\citet{eisenberg2001systemic} by introducing default costs, allowing for recovery rates strictly below one on both external and interbank assets. This extension leads to the possibility of multiple clearing solutions. Despite this added complexity, they show that the maximal proportional clearing vector can be computed efficiently. In particular, they design a polynomial-time algorithm, based on an iterative fictitious-default procedure that terminates in at most a linear number of steps, each requiring the solution of a system of linear equations. The characteristic where connections diversify risk in good times but propagate shocks in bad times was formalized by~\citet{acemoglu2015systemic} and further reviewed by~\citet{glasserman2015contagion} and~\citet{cabrales2017risk}. \citet{battiston2012debtrank} introduced DebtRank to measure the impact of distressed institutions before actual default occurs, while 
\citet{cifuentes2005liquidity} and 
\citet{shin2010risk} incorporated liquidity risk and fire sales into these contagion models. 
\citet{upper2011simulation} and 
\citet{summer2013financial} provide comprehensive surveys of these simulation methods. Further extensions by 
\citet{bardoscia2017pathways} and 
\citet{barucca2020network} have integrated valuation adjustments directly into the network dynamics.

The shift toward central clearing (the structural partner to compression) has sparked a vigorous debate on netting efficiency. 
\citet{duffie2011does}  argued that while CCPs allow for multilateral netting within a single asset class, they fracture netting sets across asset classes, potentially increasing total collateral demand. This trade-off was critically examined by~\citet{cont2014central} and~\citet{koeppl2012ccp}, who modeled the conditions under which multilateral netting dominates bilateral arrangements. 
\citet{ghamami2017does} and~\citet{duffie2015central} focused on the capital and collateral implications of this shift, questioning whether regulatory capital rules (e.g., Basel~III) actually provide the intended incentives for central clearing. The resilience of CCPs as risk concentration hubs is examined by~\citet{cont2015end} and~\citet{biais2015clearing}  who study the hierarchical structure of default resources.

Finally, the efficacy of compression is often evaluated against quantitative systemic risk measures. 
\citet{adrian2016covar} proposed CoVaR to measure the tail-risk contribution of individual institutions. 
\citet{acharya2017measuring} developed Systemic Expected Shortfall (SES), while 
\citet{brownlees2017srisk} introduced SRISK to quantify the capital shortfall of firms during a crisis. 
\citet{billio2012econometric} utilized principal component analysis to track the changing connectedness of the financial sector (hedge funds, banks, insurers) over time. This econometric branch of literature works in parallel with the network theory branch (\cite{furfine2003interbank,boss2004network,muller2006interbank}) to provide the metrics by which post-trade mechanisms like compression are ultimately judged.

\subsection{Our contribution}
In this work, we advance the algorithmic study of compression in financial networks. We consider a market composed of multiple banks, each endowed with an initial portfolio and engaged in a set of bilateral liabilities. A bank defaults if its available resources are insufficient to cover all obligations; in such cases, we consider two models. In the first model, termed \emph{proportional clearing}, the bank repays creditors proportionally to its incoming payments. In the second model, termed \emph{priority-proportional clearing}, each bank has a priority list of groups of creditors. Then, in case of default, a bank pays pays its liabilities in a sequential way to the banks in its priority groups, paying full liabilities as long as possible, and paying proportional liabilities for the last group (where payments are made). The two models are formally presented in Section~\ref{sec:preliminaries}.

Our research is motivated by the following central question: ``Can we design a compression scheme that improves outcomes for as many banks as possible in a financial network?''

We investigate this question from a computational perspective and 
obtain the following results.

In Section~\ref{sec:alg-clearing} we present a polynomial-time algorithm that computes a priority-proportional clearing for a given financial market. Moreover, our algorithm determines the coordinate-wise maximal such  clearing vector. The algorithm works by solving three carefully designed linear programs  that satisfy key properties. 

In Section~\ref{sec:np-hardness} we show that several natural questions are computationally difficult in a variety of cases. We enumerate here the hardness results. First, given an integer $k$, determining a compression such that at most $k$ banks default is an $\NP$-hard problem. The problem remains $\NP$-hard even in restricted cases when we ask for only one compression cycle or when we only allow bilateral debt clearing. 
We also show that determining if a specific bank can be saved is $\NP$-hard or if at least three banks can be saved from defaulting. These strong intractability results hold even without default costs and in the simpler proportional-clearing model.

In Section~\ref{sec:single-default} we present two polynomial-time algorithms, one for proportional clearing and one for priority-proportional clearing, which given a financial system decide whether there exists a compression  such that at most one bank defaults. The algorithms iterate over all possible banks~$b$ and test if there exists a compression in which that specific bank~$b$ defaults. Finding a compression is reduced to computing a flow with specific properties.

Finally, in Section~\ref{sec:ilp} we describe a Mixed Linear Integer Program (MILP)  to find a compression that maximizes non-defaulting banks in the priority-proportional clearing model, which gives a conceptually simple way to compute optimal solutions for these $\NP$-hard problems. We remark that even finding a maximal proportional clearing with default costs with a single integer program was not known in the literature, hence our IP is a novel contribution for both the computational and structural aspects of the problem.

\section{Preliminaries}
\label{sec:preliminaries}

    For $k\in \mathbb{N}$, let $[k]=\{ 1,\dots, k\}$ and $[k]_0=\{0,1,\dots,k\}$.
    
    A \emph{financial market} is given by a tuple $(N,L,e,\alpha, \beta)$, where $N$ is the set of \emph{banks}, $L=(L_{ij})_{i,j\in N}$ is the vector describing the \emph{liabilities} between the banks, $e=(e_i)_{i\in N}$ is the initial \emph{endowment} of each bank and 
    the vectors $\alpha,\beta \in [0,1]^N$
    are the \emph{default cost parameters}. Apart from Section~\ref{sec:alg-clearing}, we generally assume nonnegative endowments (which is a natural assumption), but allow negative ones in Section~\ref{sec:alg-clearing} for the generality of the algorithm. If $\alpha=\beta = 1^N$, we say the network \emph{has no default costs.} We assume that $L_{ii}=0$ for each $i \in N$.
    The parameter $\alpha_i$ describes the ratio up to which a defaulting bank $i\in N$ can use its own endowment, and $\beta_i$ defines the ratio up to which it can use its incoming payments.
   
    More precisely, we define the total income of 
    bank~$i$ under some payment vector~$p=(p_{ij})_{i,j \in N}$ assuming that bank~$i$ does not default (while taking into account its original endowment) as
    \[
    E_i^\nondeft(p)=e_i + \sum_{j \in N} p_{ji}.
    \]
    By contrast, the total income of
    bank~$i$ under the payment vector~$p$ and assuming that bank~$i$ \emph{does} default, which incurs the application of defaulting costs, is defined as 
    \[
    E_i^\deft(p)=\alpha_i e_i + \sum_{j \in N} \beta_i p_{ji}.
    \]

    %

    \myparagraph{Financial markets with proportional clearing.}
    In such markets, a bank either pays all its liabilities fully, or it defaults, in which case it pays all its liabilities in a proportional fashion using up all its income.
    Let $L_i:= \sum_{j\in N}L_{ij}$ denote the total liabilities that bank $i$ needs to pay.
    %
    %
    We say that a vector $p=(p_{ij})_{i,j \in N}$ is a \emph{proportional clearing vector}, if it satisfies 
    \[p_{ij} = \begin{cases}
        L_{ij} & \text{if } E_i^\nondeft(p) \geq L_i;
        \\[4pt]
        \max(0,\frac{L_{ij}}{L_i} \cdot E_i^\deft(p)) & \text{otherwise}.
    \end{cases}
    \]

    \myparagraph{Financial markets with priority-proportional clearing.}
    In such markets, each bank $i \in N$ is associated with a  \emph{priority list}
    consisting of priority groups $g_i^1,\dots, g_i^{k_i}$ where $(g_i^1,\dots, g_i^{k_i})$ is a partitioning of all banks in $N \setminus \{i\}$ with $L_{ij}>0$ (if $L_{ij}=0$, then no payment is ever needed~from $i$ to~$j$). 
    Payments then are done according to this priority ordering: first, bank~$i$ pays all its full liabilities towards banks in its first priority group $g_i^1$, then proceeds with the second group, and so on, until the moment comes when it can no longer pay all its liabilities towards the banks in the current, say the $\ell$-th, priority group (where $1 \leq \ell \leq k_i$. In that case, it uses up its remaining income to pay its liabilities towards the banks in the priority group $g_i^\ell$ in a proportional fashion, and ignores all banks in lower priority groups.
    
    We will write $G_i^{\ell} = \cup_{k=1}^{\ell} g_i^k$  to denote the set of banks having priority at most~$\ell$; for simplicity, we set
    $G_i^0:= \emptyset$.
    Let $L_i^{\ell} = \sum_{j\in g_i^{\ell}}L_{ij}$ denote the liabilities owned by bank~$i$ to the banks belonging to its $\ell$-th priority group. 
    %
    We say that a vector $p$ is a \emph{priority-proportional clearing vector} if for each bank~$i\in N$ and $j\in g_i^{\ell}$ it satisfies one of the two conditions:
    \begin{itemize}
        \item Bank~$i$ does not default and pays all its liabilities, 
        meaning that $E_i^\nondeft(p) \geq \sum_{j \in N} L_{ij}$ and thus $i$ can pay its liabilities fully, i.e., $p_{ij}=L_{ij}$ for all $j \in N$.
        \item Bank~$i$ defaults and pays its liabilities in a sequential way to the banks in its priority groups, paying full liabilities as long as possible, and paying proportional liabilities for the last group where payments are made. 
        That is, $E_i^\nondeft(p) < \sum_{j\in N} L_{ij}$ and the payments by~$i$ towards some bank~$j$ in the $\ell$-th priority group (i.e., with $j\in g_i^{\ell}$) are defined as
        \begin{align*}
        p_{ij} &= 
        \begin{cases}
        L_{ij} & \text{if } E_i^\deft(p) 
        \geq 
        \sum_{h=1}^\ell L_i^h;
        \\[4pt]
        \frac{L_{ij}}{L_i^{\ell}}
        (E_i^\deft(p) - \sum_{h=1}^{\ell-1}L_i^h) & \text{if }
        \sum_{h=1}^{\ell-1}L_i^h  \leq E_i^\deft(p)<\sum_{h=1}^\ell L_i^h;
        \\[4pt]
        0 & \text{otherwise.}
        \end{cases}
        \end{align*}
    \end{itemize}


    We formally define the problem of computing a priority-proportional clearing vector as follows.

    \medskip
\noindent
\begin{center}
\begin{minipage}{\textwidth}
\fbox{
    \begin{tabular}{@{\hspace{2pt}}l@{\hspace{4pt}}p{0.84\textwidth}@{\hspace{2pt}}}%
        \multicolumn{2}{@{\hspace{2pt}}p{0.87\textwidth}}{
        \PriPropClear
        } \\
        {\bf Input:} &
            A financial market $M=(N,L,e,\alpha,\beta)$. 
        \\
        {\bf Task:} &
            Compute a priority-proportional clearing vector for~$M$.
    \end{tabular}
    }
\end{minipage}
\end{center}
\medskip

    \myparagraph{Compressions.}
    We say that a set $i_0,i_1,\dots, i_{\ell-1}$ of $\ell$ banks form a \emph{liability cycle} if $L_{i_j,i_{j+1 \bmod \ell}}>0$ for each $j \in \{0,1,\dots,\ell-1\}$. 
    In other words, a liability cycle is a cycle in the directed graph underlying the market, i.e, the graph over $N$ where $(i,j)$ is an arc if and only if $L_{ij}>0$.
    If banks $i_0,i_1,\dots,i_{\ell-1}$ form a liability cycle, then they may decide to decrease their liabilities along this cycle by some amount $\varepsilon \leq \min \{L_{i_j,i_{j+1 \bmod \ell}}: j \in \{0,1,\dots,\ell-1\}\}$.
    Such an operation can be performed simultaneously for several, possibly overlapping, liability cycles. This concept is described by the notion of a compression.
    
    Formally, a \emph{compression} $C=(C_{ij})_{i,j \in N}$ is a vector that satisfies $0 \leq C_{ij} \leq L_{ij}$ for each $i,j \in N$, as well as the flow condition 
    $\sum_{j \in N} C_{ji}=\sum_{j \in N} C_{ij}$; that is, a compression can be considered as the union of liability cycles.
    Given a financial market $M=(N,L,e,\alpha,\beta)$ and a compression $C$ for it, the financial market $M-C$ obtained after \emph{applying compression~$C$} is the market $(N,L',e,\alpha,\beta)$ where the new liability vector $L'$ is defined as $L'_{ij}=L_{ij}-C_{ij}$ for each $i,j \in N$.

\begin{figure*}[t!]
    \centering
    \begin{subfigure}[t]{0.2\textwidth}
        \centering
        \includegraphics[scale=0.75]{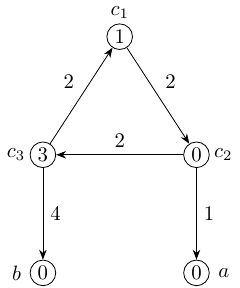}
        \caption{Market~$M$.}
        \label{fig:ex-market}
    \end{subfigure}%
    \hspace{20pt}
    \begin{subfigure}[t]{0.2\textwidth}
        \centering
        \includegraphics[scale=0.75]{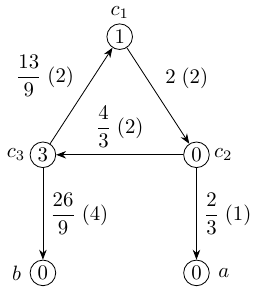}
        \caption{Clearing vector for the original market~$M$ without compression.}
        \label{fig:ex-comp-0}
    \end{subfigure}
    \hspace{23pt}
    \begin{subfigure}[t]{0.2\textwidth}
        \centering
        \includegraphics[scale=0.75]{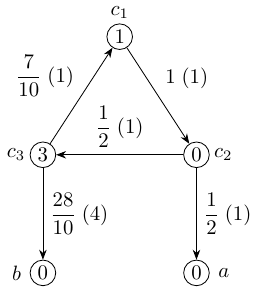}
        \caption{Clearing vector after applying compression with value~$\varepsilon=1$.}
        \label{fig:ex-comp-1}
    \end{subfigure}
    \hspace{20pt}
    \begin{subfigure}[t]{0.2\textwidth}
        \centering
        \includegraphics[scale=0.75]{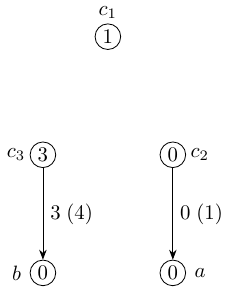}
        \caption{Clearing vector after applying compression with value~$\varepsilon=2$.}
        \label{fig:ex-comp-2}
    \end{subfigure}
    \caption{Illustration for Example~\ref{ex}. Liability values are displayed alongside the corresponding arrows, shown in parenthesis for Figures (b)--(d). Initial endowments are depicted within the circles representing the banks. Figures~(b), (c), and~(d) show clearing vectors where the compression value along the cycle $(c_1,c_2,c_3)$ is $\varepsilon=0,1,2$, respectively, with the payment values written before the liabilities (after compression) in parenthesis.
    }
\end{figure*}

The following simple example illustrates how applying a compression can have different effects on different banks in the market.

\begin{example}
\label{ex}
Consider a financial market~$M$ with banks $N=\{a,b,c_1,c_2,c_3\}$ as the set of banks where the network of liabilities and the initial endowments are as displayed in Figure~\ref{fig:ex-market}.
We assume that $\alpha_i=\beta_i=1$ for all banks $i \in  N$.
Figure~\ref{fig:ex-comp-0} displays a clearing vector for market~$M$. It can be seen that the defaulting banks are~$c_2$ and~$c_3$.

There is only a single cycle in the directed graph underlying~$M$, going through banks~$c_1$, $c_2$, and~$c_3$. Hence, all compressions have the form of such a cycle with some compression value~$\varepsilon \in [0,2]$.

Figures~\ref{fig:ex-comp-1} and~\ref{fig:ex-comp-2} display the markets obtained by applying such a compression with value~$\varepsilon=1$ and~$\varepsilon=2$, respectively, together with a corresponding clearing vector.

Let us consider the effects of the compression on banks~$a$ and~$b$. The (total) payments they recieve in the markets after applying a compression with value~$\varepsilon=0,1,2$ are as follows:

\begin{center}
\begin{tabular}{lccc}
& $\varepsilon=0$ & $\varepsilon=1$ & $\varepsilon=2$
\\
\hline
\\[-8pt]
bank $a$: 
	& $\frac{2}{3}=0.6\dot{6}$ 
	& $\frac{1}{2}=0.5$
	& $0$
\\[4pt]
bank $b$: 
	& $\frac{2}{3}=2.8\dot{8}$ 
	& $\frac{28}{10}=2.8$
	& $3$
\end{tabular}
\end{center}

Hence, it can be seen that while bank~$a$ prefers not to apply any compression, bank~$b$ does benefit from applying a compression with value~$\varepsilon=2$ (completely eliminating the cycle of liabilities in the market), but does not benefit from applying a compression with smaller value $\varepsilon=1$. 
\end{example}

The main optimization problem we are interested in is to find a compression whose application yields a clearing vector minimizing the number of defaulting banks. 
We formalize this question as a decision problem as follows.

\medskip
\noindent
\begin{center}
\begin{minipage}{\textwidth}
\fbox{
    \begin{tabular}{@{\hspace{2pt}}l@{\hspace{4pt}}p{0.84\textwidth}@{\hspace{2pt}}}%
        \multicolumn{2}{@{\hspace{2pt}}p{0.87\textwidth}}{
        \FindComp\
        } \\
        {\bf Input:} &
            A financial market $M=(N,L,e,\alpha,\beta)$ and an integer~$k$. 
        \\
        {\bf Question:} &
            Is there a compression~$C$ for~$M$ such that the market~$M-C$ admits a proportional clearing vector under which at most~$k$ banks default?
    \end{tabular}
    }
\end{minipage}
\end{center}
\medskip

We will also investigate variants of the above problem where the desired compression needs to satisfy certain restrictions.  
A compression~$C$ for $M$ \emph{is bilateral} if $C_{ij}=C_{ji}$ for each pair $i,j \in N$ of banks. 
Moreover, $C$ is a \emph{cycle compression} if $C$ is, roughly speaking, the union of disjoint cycles; formally, if $C_{ij} > 0$ for some $i,j \in N$ implies $C_{ij'}=0$ for all $j' \in N \setminus \{j\}$.
Searching for bilateral or cycle compressions is motivated by their simplicity, possibly allowing for faster optimization.

\medskip
\noindent
\begin{center}
\begin{minipage}{\textwidth}
\fbox{
    \begin{tabular}{@{\hspace{2pt}}l@{\hspace{4pt}}p{0.84\textwidth}@{\hspace{2pt}}}%
        \multicolumn{2}{@{\hspace{2pt}}p{0.87\textwidth}}{
        \textsc{Optimal Bilateral ({\normalfont{or }} Cycle) Compression for Proportional Clearing}
        } \\
        {\bf Input:} &
            A financial market $M=(N,L,e,\alpha,\beta)$ and an integer~$k$. 
        \\
        {\bf Question:} &
            Is there a bilateral (or cycle, respectively) compression~$C$ for~$M$ 
         such that $M-C$ admits a proportional clearing vector under which at most~$k$ banks default?
    \end{tabular}
    }
\end{minipage}
\end{center}

\medskip
A further very natural and interesting question is one where we ask whether a certain fixed bank can be saved from bankruptcy.
To define the corresponding problem precisely,  let $\dft^M(p)$ denote the banks defaulting under some clearing vector~$p$ for~$M$. 

\medskip
\noindent
\begin{center}
\begin{minipage}{\textwidth}
\fbox{
    \begin{tabular}{@{\hspace{2pt}}l@{\hspace{4pt}}p{0.84\textwidth}@{\hspace{2pt}}}%
        \multicolumn{2}{@{\hspace{2pt}}p{0.87\textwidth}}{
        \FindCompBank
        } \\
        {\bf Input:} &
            A financial market $M=(N,L,e,\alpha,\beta)$ and a bank $b \in N$. 
        \\
        {\bf Question:} &
            Is there a compression $C$ for $M$ and a proportional clearing vector  for the market $M-C$ under which $b$ does not default?
    \end{tabular}
    }
\end{minipage}
\end{center}
\medskip

\section{Algorithm for Priority-Proportional Clearing}
\label{sec:alg-clearing}

In this section, we prove the following result.

\begin{theorem}
\label{thm:alg-pripropclear}
    \PriPropClear\ can be solved in polynomial time. 
\end{theorem}

We provide an algorithm that computes a coordinate-wise maximal priority-proportional clearing vector. The core difficulty lies in the fact that the presence of the default cost parameters $\alpha_i, \beta_i$ depends on whether bank $i$ defaults, and the payments bank $i$ receives also depend on whether the other banks default, so it is not straightforward to determine whether $i$ will default or not.

Our strategy is iterative. We maintain a set of ``active'' assumptions about which banks default and which priority groups are fully paid. To avoid case distinction, we introduce parameters $\alp,\bet\in [0,1]^N$, where we set $(\alp_i,\bet_i)=(\alpha_i,\beta_i)$ if $i$  defaults, and $(\alp_i,\bet_i)=(1,1)$ otherwise.

For a fixed set of parameters $\alp, \bet$, we use Linear Programming (LP). If the result contradicts our assumptions (e.g., a bank assumed solvent actually defaults), we update the parameters and repeat.

\subsection{Clearing with Fixed Parameters}

First, consider the subproblem where the parameters $\alp_i, \bet_i$ are fixed constants for each $i \in N$ (i.e., we tentatively assume we know which banks default). Furthermore, we also tentatively assume that for each bank~$i$ we know the \emph{critical priority group } $\Gamma_i=g_i^{\gamma_i}$, that is, the highest priority group towards which $i$ can pay a non-zero amount.
Initially, we will assume that $(\alp_i,\bet_i)=(1,1)$ and $\gamma_i=k_i$ for all $i\in N$, and then update our beliefs when necessary. 



\myparagraph{Iterative construction of the sequence $(p^k,z^k)_{k\in\mathbb{N}}$.}
We are going to define an iterative sequence of payment--subsidy pairs
\[
(p^k,z^k)_{k\in\mathbb{N}} \qquad \text{where } 
p^k=(p_{ij}^k)_{i,j\in N}
\text{ and } 
z^k=(z_i^k)_{i\in N}
\]
To ensure that the payments respect the  priority structure and our fixed choice for the critical priority groups,
for all $k \in \mathbb{N}$ we  
set the payment~$p^k_{i,j}$ form some bank~$i \in N$  towards a bank~$j \notin \Gamma_i$ that is not in the critical priority group for~$i$ as
\[p_{ij}^k= \begin{cases}
   L_{ij} & \text{if } j\in g_i^{\ell} \text{ for some } l<\gamma_i; \\
   0 & \text{if } j\in g_i^{\ell} \text{ for some } \ell>\gamma_i.
\end{cases}
\]

Next, we define the iteration that defines the payment $p^k_{i,j}$ for some banks $i \in N$ towards a bank~$j \in \Gamma_i$ in the critical priority group, as well as the subsidies $z^k_i$ as follows.

\myparagraph{Initialization:}
At iteration $k=0$, we set
\[
p_{ij}^0 := L_{ij}
\quad \text{for all } i\in N \text{ and } j\in \Gamma_i,  \text{ and }
\qquad
z_i^0 := 0 \quad \text{for all } i\in N .
\]

\myparagraph{Update rule:} 
Suppose $(p^k,z^k)$ is given. 
We define the \emph{net equity} for each bank~$i \in N$ as
\[
\EE_i(p^k,z^k)
\;:=\;
\alpha_i e_i
\;+\;
z_i^k
\;+\;
\sum_{j\in N}\bigl(\beta_i p_{ji}^k - p_{ij}^k\bigr).
\]
Then, for each $i\in N$, we determine  $(p_i^{k+1},z_i^{k+1})$ uniquely as follows:
\begin{itemize}
    \item If $\EE_i(p^k,z^k)\ge 0$, then we set
    \[
    (p_i^{k+1},z_i^{k+1}) := (p_i^k,z_i^k).
    \]
    \item If $\EE_i(p^k,z^k)<0$, then 
    we let
\begin{align*}
p_{ij}^{k+1}& =\max \left\{ 0, p_{ij}^k+ \EE_i(p^k,z^k)\frac{L_{ij}}{
L_i^{\Gamma_i}
} \right\} \qquad \text{for each } j\in \Gamma_i, \text{ and}
\\
    z_i^{k+1} &= \max \left\{ 0, \sum_{j\in N} p_{ij}^{k+1}
-    \alpha_i e_i -
    \sum_{j\in N}\beta_i p_{ji}^k\right\}.
\end{align*}

\end{itemize}

\medskip

Note that at each iteration $k$ the priority-proportionality structure is preserved: 
\[
p_{ij}^k=\begin{cases}
    L_{ij} & \text{if } j\in g_i^{\ell} \text{ with } \ell<\gamma_i, \\
    \lambda_i^{\gamma_i}L_{ij}
    & \text{if } j\in \Gamma_i, \\ 
    0 &  \text{if } j\in g_i^{\ell} \text{ with } \ell>\gamma_i.
\end{cases}
\]
for some $\lambda_i^{\gamma_i}\in [0,1]$
for each $i\in N$.

The procedure is iterated until
$\EE_i(p^k,z^k)\ge 0$ for all $i\in N$.

\myparagraph{Monotonicity and convergence.}
The sequence $(p^k,z^k)$ has the following properties:
\begin{itemize}
    \item for each $i\in N$, the sequence $(z_i^k)_{k\in\mathbb{N}}$ is monotone nondecreasing;
    \item for each $i,j\in N$, the sequence $(p_{ij}^k)_{k\in\mathbb{N}}$ is monotone nonincreasing.
\end{itemize}
Moreover, the sequences are bounded: $p_{ij}^k\ge 0$ for all $i,j,k$, and
$z_i^k\le L_i-e_i$ for all $i,k$.  
Consequently, by the Monotone-convergence theorem, the sequence
$(p^k,z^k)$ admits a limit point, and hence converges.

    By the construction, the limit point $(p^\star, z^\star)$ satisfies the following property.
    
    \begin{itemize}
        \item[$(\mypropsymbol)$] \raisebox{\ht\strutbox}{\hypertarget{prop:spec}{}}  If $z_i^\star>0$, then $p_{ij}^\star=0$ for all $j\in \Gamma_i$ and $\EE_i(p^\star,z^\star)=0$.
    \end{itemize}

\def\specprop{\hyperlink{prop:spec}{($\mypropsymbol$)}}

Intuitively, 
\specprop\  means that if a bank requires external subsidy ($z_i^\star > 0$),
it cannot afford to make any payments to its critical priority group.

\subsubsection{Finding the limit point}
Consider the following linear program, denoted as \ref{prob:min_subsidy}:
\begin{equation}
\tag{\textsc{LP:Min-Subsidy}}
\label{prob:min_subsidy}
\begin{array}{rlclll}
    \min & \multicolumn{4}{l}{\sum\limits_{i \in N} z_i} \\[4pt]
    \text{s.t.}
    & \sum_{j\in N}p_{ij} & \leq & \alp_i e_i + z_i + \sum_{j\in N}\bet_{i}p_{ji} & \forall i \in N & \text{(Budget)} \\[2pt]
    & p_{ij} & = & \lambda_i^{\ell}L_{ij} & \forall i \in N, j\in g_i^{\ell} \quad & \text{(Prop)} \\[2pt]
    & \lambda_i^{\ell}, z_i & \ge & 0 & \forall i\in N, \ell \in [k_i] \\[2pt]
    & \lambda_i^{\gamma_i} & \le & 1 & \\[2pt]
    & \lambda_i^{\ell} & = & 1 & \forall i\in N,\ell <\gamma_i\\[2pt]
    & \lambda_i^{\ell} & = & 0 & \forall i\in N, \ell >\gamma_i
\end{array}
\end{equation}

We will show that the optimal solution to \ref{prob:min_subsidy} gives the minimum total subsidy 
$Z^* = \sum z_i$ required. However, there may be multiple solutions with the same 
minimal $Z^*$. We are interested in a specific solution that satisfies property \specprop.


To find a solution satisfying \specprop, we solve a second linear program, \ref{prob:min_payment}. Let $Z^*$ be the optimal objective value of \ref{prob:min_subsidy}. While this counterintuitively minimizes 
the total payment, we show that it does help us compute the $z_i$ values of the optimal solution.

\begin{equation}
\tag{\textsc{LP:Min-Payment}}
\label{prob:min_payment}
\begin{array}{rlcll}
    \min & \multicolumn{4}{l}{\sum\limits_{i,j \in N} p_{ij}} \\[4pt]
    \text{s.t.}
    & \sum_{i \in N} z_i  & = & Z^* & \\[2pt]
    & \multicolumn{4}{l}{\text{All constraints from } \text{\ref{prob:min_subsidy}}.}
\end{array}
\end{equation}

\begin{restatable}[\linkproof{app:prf-lemsol}]{lemma}{lemsol}
\label{lem:sol}
    The limit point $(p^\star,z^\star)$ is an optimal solution to \ref{prob:min_subsidy}.
\end{restatable}

\begin{restatable}[\linkproof{app:prf-lemuniq}]{lemma}{lemuniq}
\label{lem:uniq}
    Among the optimal solutions of \ref{prob:min_subsidy} that satisfy property~\specprop, the values of $z_i$ are unique.
\end{restatable}

\begin{restatable}[\linkproof{app:prf-lemminp}]{lemma}{lemminp}
\label{lem:minp}
    An optimal solution to \ref{prob:min_payment} corresponds to an optimal solution of \ref{prob:min_subsidy} that satisfies property~\specprop.
\end{restatable}

By Lemmas~\ref{lem:sol}-\ref{lem:minp}, we can find the exact $z_i^\star$ values of the limit point.
Finally, we demonstrate how to compute $p^\star$ in polynomial time as well.
Take
\begin{equation}
\tag{\textsc{LP:Max-Payment}}
\label{prob:max_payment}
\begin{array}{rlcll}
    \max & \multicolumn{4}{l}{\sum\limits_{i,j \in N} p_{ij}} \\[4pt]
    \text{s.t.}
    &  z_i  & = & z_i^\star
& \forall i\in N \\[2pt]
& \sum_{j\in \Gamma_i}p_{ij}& = & 0 & i\in N: z_i^\star >0\\ 

& \sum_{j\in N}p_{ij} &=& \alp_ie_i+z_i^\star + \sum_{j\in N}\bet_ip_{ji} & i\in N:z_i^\star >0\\
    & \multicolumn{4}{l}{\text{All constraints from } \text{\ref{prob:min_subsidy}}.}
\end{array}
\end{equation}

\begin{restatable}[\linkproof{app:prf-lemzero}]{lemma}{lemzero}    
\label{lem:zero}
  The limit point $(p^\star,z^\star)$ is the unique optimal solution to \ref{prob:max_payment}, and is coordinate-wise maximal in $p$ among all solutions to \ref{prob:max_payment}.
\end{restatable}

\subsection{The Complete Algorithm}

We now combine the LP approach with the update strategy for the default parameters.
\begin{algorithm}[t]
\caption{Priority-Proportional Clearing}
\label{alg:priority_proportional_clearing}

\ForEach{$i \in N$}{
    Set initial values:  $\;(\alpha_i^0,\beta_i^0) \gets (1,1)$, 
    $\;\gamma_i^0 \gets k_i$, $\;$ and 
    $\;\Gamma_i^0 \gets \{ g_i^{k_i} \}$\;
}

\For{$r\gets0$ \KwTo $\infty$}{
    Set $(\alp_i^{r+1},\bet_i^{r+1},\gamma_i^{r+1},\Gamma_i^{r+1})\leftarrow (\alp_i^{r},\bet_i^{r},\gamma_i^{r},\Gamma_i^{r})$ \quad for each $i\in N$\;
    $\mathrm{Changed \gets False}$\;
    Solve \ref{prob:min_subsidy} to obtain $Z^r$\;
    Solve \ref{prob:min_payment} with $\sum_{i\in N} z_i = Z^r$ to obtain $z^r$\;
    Solve \ref{prob:max_payment} according to $z^r$ to obtain $p^r$\;

    \ForEach{$i \in N$}{
        \If{$\alpha_i^r e_i +  \sum_{j\in N} \beta_i^r p_{ji}^r < L_i$
        \textup{\textbf{and}} $(\alpha_i^r,\beta_i^r)=(1,1)$}{
            Set $\;(\alpha_i^{r+1},\beta_i^{r+1}) \gets (\alpha_i,\beta_i)$\;
            $\mathrm{Changed \gets True}$\;
        }

        \If{$z_i^r > 0$ \textup{\textbf{and}} $\gamma_i^r>0$}{
            Set $\;\gamma_i^{r+1} \gets \gamma_i^r-1\;$ and  
            $\;\Gamma_i^{r+1} \gets \{ g_i^{\gamma_i^{r+1}} \}$\;
        $\mathrm{Changed \gets True}$\;
        }
        
    }
    \If{$\mathrm{Changed=False}$}{ \Return $p^r$\;}
}
\end{algorithm}

\begin{theorem}
 There always exists a priority-proportional clearing vector. Furthermore, Algorithm~\ref{alg:priority_proportional_clearing} finds a priority-proportional clearing vector $p$ that is coordinate-wise maximal in polynomial time.
\end{theorem}
\begin{proof}
\myparagraph{Correctness.} 
 By Lemmas~\ref{lem:sol}-\ref{lem:zero}, we know that we can find the limit points $(p^r,z^r)$ by solving three linear programs \ref{prob:min_subsidy}, \ref{prob:min_payment}, \ref{prob:max_payment}, given fixed values for $(\alp_i^r,\bet_i^r),\gamma_i^r$ for $i\in N$.

    Let $p$ be a priority-proportional clearing for the input market. Set $(\alp_i,\bet_i)=(1,1)$ if $i$ is solvent under~$p$, and $(\alp_i,\bet_i)=(\alpha_i,\beta_i)$ if $i$ defaults. Also, set $\gamma_i$ as the highest priority group such that $\sum_{j\in G_i^{\gamma_i-1}}L_{ij}\le \alp e_i + \sum_{j\in N}\bet_ip_{ji}$, that is, the critical priority group for~$i$, if there exists such a value~$\gamma_i$; otherwise, i.e., if $\alp e_i + \sum_{j\in N}\bet_ip_{ji}<0$, set $\gamma_i = 0$.  

    \myparagraph{Convergence.}
    We claim that for any iteration $r$, we have that $(p,\alp,\bet,\gamma) \le (p^r,\alp^r,\bet^r,\gamma^r)$ coordinate-wise.
 In order to prove this, it suffices to show that $(p,\alp,\bet,\gamma) \le (p^{r,k},\alp^r,\bet^r,\gamma^r)$, where $(p^{r,k},z^{r,k})$ is the $k$-th term in the series converging to the limit point for the given values $\alp^r,\bet^r$, and $\gamma^r$.

    We show this by lexicographic induction on $(r,k)$, where $r$ denotes the iteration step within the iteration for computing the defaulting parameters and critical priority groups, and $k$ is the iteration step in the series for computing the payment--subsidy values for some fixed defaulting parameters and critical priority groups. For $r=0$, $(\alp,\bet,\gamma) \le (\alp^0,\bet^0,\gamma^0)$ is trivial, as each is set to the maximum possible value. Also, $p_{ij}\le p_{ij}^{0,0}$, as $p_{ij}^{0,0}=L_{ij}$ for $i,j\in N$. Assume we know $p_{ij}\le p_{ij}^{0,k}$ and take $p_{ij}^{0,k+1}$. If $p_{ij}^{0,k+1}\ne p_{ij}^{0,k}$, then $\sum_{j\in N}p_{ij}^{0,k+1}=  \alp_i^0e_i+z_i^{0,k+1}+\sum_{j\in N}\bet_i^0p_{ji}^{0,k}\ge \alp_ie_i+\sum_{j\in N}\bet_ip_{ji}\ge \sum_{j\in N}p_{ij}$ using $z_i^{0,k+1}\ge 0$. By priority-proportionality, we get that $p_{ij}\le p_{ij}^{0,k+1}$ for any $i,j\in N$ as desired. Hence, for the limit $p_{ij}^0$, we also have $p_{ij}\le p_{ij}^0$ for $i,j\in N$.

    Now, suppose we know that the claim holds for $r$ and every $k\in \mathbb{N}$, and take $r+1$. If $(\alp_i^{r+1},\bet_i^{r+1})\ne (\alp_i^r,\bet_i^r)$ for some $i\in N$, then we have $(\alp_i^r,\bet_i^r)=(1,1)$, $\gamma_i^r=k_i$ and $\alp_i^re_i+z_i^r+\sum_{j\in N}\bet_i^rp_{ji}^r<L_i$ 
    which implies $\alp_ie_i + \sum_{j\in N}\bet_ip_{ji}\le\alp_i^re_i+\sum_{j\in N}\bet_i^rp_{ji}^r<L_i$ by induction. Hence, $i$ cannot receive enough money under~$p$ to pay all its debts ($L_i$), hence it defaults, meaning $(\alp_i,\bet_i) =( \alpha_i,\beta_i)\le (\alp_i^{r+1},\bet_i^{r+1})$. 

    Similarly, if $\gamma_i^r\ne  \gamma_i^{r+1}$, then  we have $z_i^{r}>0$, $\sum_{j\in \Gamma_i^r}p_{ij}^{r}=0$ (where $\Gamma_i^r=g_i^{\gamma_i^{r}}$), and $\sum_{j\notin G_i^{\gamma_i^r}}p_{ij}^{r}=0$, therefore $\alp_ie_i + \sum_{j\in N}\bet_ip_{ji} -\sum_{j\in G_i^{\gamma_i^r-1}}L_{ij} \le\alp_i^re_i+\sum_{j\in N}(\bet_i^rp_{ji}^r-p_{ij}^{r})= -z_i^{r}<0$, by induction and property~\specprop\  of $(p^r,z^r$). Hence, $\alp e_i + \sum_{j\in N}\bet_ip_{ji}<\sum_{j\in G_i^{\gamma_i^r-1}}L_{ij}$, thus we get $\gamma_i\le \gamma_i^{r+1}=\gamma_i^r-1$. 

    Finally, we show that $p_{ij}\le p_{ij}^{r+1,k}$ for any $k\in \mathbb{N}$. Let $k=0$. Then, $p_{ij}^{r,0}=L_{ij}\ge p_{ij}$, whenever $j\in G_i^{\gamma_i^r}$ and otherwise by induction $0=p_{ij}=p_{ij}^{r+1,0}$. Suppose we know the statement for $(r+1,k)$ and consider $(r+1,k+1)$. 
    If $p_{ij}^{r+1,k+1}\ne p_{ij}^{r+1,k}$, then $\sum_{j\in N}p_{ij}^{r+1,k+1}=  \alp_i^{r+1}e_i+z_i^{r+1,k+1}
    +\sum_{j\in N}\bet_i^{r+1}p_{ji}^{r+1,k}\ge \alp_ie_i+\sum_{j\in N}\bet_ip_{ji}\ge \sum_{j\in N}p_{ij}$ using $z_i^{r+1,k+1}\ge 0$. By priority-proportionality, we get that $p_{ij}\le p_{ij}^{r+1,k+1}$ too for any $i,j\in N$ as desired. Hence, for the limit $p_{ij}^{r+1}$, we also have $p_{ij}\le p_{ij}^{r+1}$ for $i,j\in N$.

    \myparagraph{Output.}
    The algorithm terminates once it reaches a point where $z^r_i>0$ implies $\gamma_i^r=0$ for each bank~$i$, and all banks with $(\alp_i^r,\bet_i^r)=(1,1)$ are able to pay their liabilities fully. If $(\alp_i^r,\bet_i^r)=(\alpha_i,\beta_i)$, indicating that $i$ defaults, and  $E_i^\deft (p^r)= \alp_i^r e_i + \sum_{j\in N}\bet_i^rp_{ji}^r<0$, then $z_i^r>0$ and $p_{ij}^r=0$ for all $j\in N$, as required by priority-proportionality for this special case.
    Also, when $\alp_i^r e_i + \sum_{j\in N}\bet_i^rp_{ji}^r\ge 0$, then $z_i^r=0$ (if $z_i^r>0$, then $\sum_{j\in N}p_{ij}^r=0$ by \ref{prob:max_payment}, so $z_i^r$ could be decreased while still remaining a solution to \ref{prob:min_subsidy}, a contradiction), so the payments are priority-proportional by the constraints of \ref{prob:max_payment} and the constraints~(Prop) in \ref{prob:min_subsidy}. As we have shown above, the output $p^r$ from the limit point $(p^r,z^r)$ is a coordinate-wise upper bound on $p$ for any priority-proportional clearing vector $p$, so $p^r$ must be a coordinate-wise maximal priority-proportional clearing vector.

    
    \myparagraph{Running time.} In every iteration, either a bank's parameters $(\alp_i^r,\bet_i^r)$ change from solvent $(1,1)$ to default $(\alpha_i,\beta_i)$, or a priority layer is changed ($\gamma_i^r$ decreases). Since default parameters only switch at most once and the number of priority layers is finite, the algorithm terminates. Specifically, it terminates in at most $O(|N|\sum_{i \in N} k_i)$ iterations. Since each step involves solving polynomial-size LPs, the total running time is polynomial.
\end{proof}

For the case when the market uses proportional clearing, our algorithm can be substantially simplified; see Appendix~\ref{app:simplify} for details.

\section{Computational Intractability of Finding an Optimal Compression}
\label{sec:np-hardness}

Let us now present our results regarding the intractability of finding an optimal compression, i.e., one that saves as many banks as possible from defaulting.
The following results are proved via polynomial-time reductions form the classic $\NP$-hard problem \maxsat. 

\begin{restatable}[\linkproof{app:prf-thmnphfindcomp}]{theorem}{thmnphfindcomp}
\label{thm:nph:findcomp}    
Each of the following problems is $\NP$-hard even with no default costs and nonnegative endowments:
\begin{itemize}
    \item[(a)] \FindComp;
    \item[(b)] \FindCompPairs;
    \item[(c)] \FindCompCycle.
\end{itemize}
     
\end{restatable}

Interestingly, deciding whether it is possible to find a compression that enables us to save \emph{all but three} banks from defaulting is computationally hard. 
In fact, even deciding whether a single, fixed bank can be saved from bankruptcy is $\NP$-hard.
These strong results are proved via a reduction from the well-known \partition\ problem.

\begin{restatable}[\linkproof{app:prf-thmnphallbutthree}]{theorem}{thmnphallbutthree}
\label{thm:nph:all-but-three}
    \FindComp\ is $\NP$-hard, even if $k=|N|-3$. Furthermore, \FindCompBank\ is $\NP$-hard. Both hardness results hold without default costs and nonnegative initial endowments.
\end{restatable}

\section{Finding a Compression to Save One or All but One Banks}
\label{sec:single-default}

    First, we show that we can always prevent the situation that \emph{all} banks default: a greedy strategy for finding compression cycles is already sufficient to guarantee that at least one bank does not default. 
    \begin{proposition}
        \label{prop:greedy-saves-one-bank}
        When endowments are nonnegative, there is a polynomial-time that given a financial market~$M$ with $n$ banks finds a compression~$C$ for~$M$ such that $I-C$ admits a proportional clearing vector~$p$ with $|\dft^{M-C}(p)|\leq n-1$. 
    \end{proposition}

    \begin{proof}
        We create the desired compression~$C$ in a greedy manner. We start by setting~$C$ to be zero on all arcs of the digraph~$D=(N,A)$ underlying~$M$. Then, we repeat the following step: we find an arbitrary cycle~$K$ in~$D$, determine the maximum current liability along~$K$, say~$\varepsilon_K$, and compress along $K$ with value~$\varepsilon_K$. Formally, this means that we reduce each liability along~$K$ by~$\varepsilon_K$, increase the value of~$C$ along these arcs by~$\varepsilon_K$, 
        and update the graph~$D$ accordingly; note that this update will result in a digraph with strictly fewer arcs than~$D$. 
        
        Repeating this step as long as possible, 
        we end up with a compression~$C$ for which the graph~$D'$ underlying $M-C$ is acyclic. Then, $D'$ contains at least one sink, i.e., a bank~$b$ with no outgoing arc in~$D'$.
        Since $e_b\ge 0$, it will not default under any clearing vector for~$M-C$. 

        Since finding a cycle in a digraph can be done in linear time, and the number of iterations is at most~$|A|$, the running time of this algorithm is clearly polynomial in the size of~$M$.
    \end{proof}

    In the rest of this section, we deal with the problem that can be viewed as the \emph{dual} of the problem solved by Proposition~\ref{prop:greedy-saves-one-bank}: given a financial market, can we find a compression that enables us to save \emph{all but one} banks from defaulting? 
    As we will see, solving this problem requires advanced techniques and quite technical arguments. 
    To ease the presentation and understanding of our algorithm, we first explain it for the simpler model that uses proportional clearing in Section~\ref{sec:algo:save-all-but-one:proportional}, and 
    then present its extension for the model using priority-proportional clearing in Section~\ref{sec:algo:save-all-but-one:priority-proportional}.    

\subsection{Saving all but one bank under proportional clearing}    
    \label{sec:algo:save-all-but-one:proportional}

    This section is dedicated to proving the following result.
    \begin{theorem}
    \label{thm:allbutone}
        Given a financial market~$M$, one can decide in polynomial time 
        whether $M$ admits a compression~$C$ such that $M-C$ admits a proportional clearing vector~$p$ with $|\dft^{M-C}(p)|\leq 1$.
    \end{theorem}
    
    \begin{proof}
    Let $D=(N,A)$ denote the underlying digraph of~$M$ where $(i,j) \in A$ if and only if $L_{ij}>0$. For a set~$F$ of arcs in~$D$, let $\chi^F \in \{0,1\}^{N \times N}$ denote the characteristic vector of~$F$, i.e., $\chi^F_{ij}=1$ if and only if $(i,j) \in F$ and $\chi^F_{ij}=0$ otherwise.
    The following observation is  straightforward to verify. \begin{observation}
    \label{obs:remove-cycles}
        Let $C$ be a cycle in~$D$ such that there is a clearing vector~$p$ for~$M$ for which no bank appearing on~$C$ defaults under~$p$, and let $\varepsilon=\min_{(i,j) \in C}L_{ij}$. 
        Then $p'=p-\varepsilon \cdot \chi^C$ is a clearing vector for~$M-C$ with $\dft^M(p)=\dft^{M-C}(p')$. 
        Conversely, if~$q$ is a clearing vector for~$M-C$ such that no bank on~$C$ defaults under~$q$, then $q'=q+\varepsilon \cdot \chi^C$ is a clearing vector for~$M$ with $\dft^M(q')=\dft^{M-C}(q)$.
    \end{observation}

    Observation~\ref{obs:remove-cycles} implies  that there is a compression where all banks remain solvent if and only if all banks remain solvent under a clearing vector for the original market~$M$. As we can efficiently decide whether this is the case, from now on we assume that $b$ defaults after any compression.

\medskip
    We start with guessing the single bank~$b$ for which some compression~$C$ yields a clearing vector~$p$ for~$M-C$ with $\dft^{M-C}(p) \subseteq \{b\}$. 
    We set $\widetilde{N}=N \setminus \{b\}$, and we call a compression~$C$ \emph{suitable} if $M-C$ admits a clearing vector~$p$ with $\dft^{M-C}(p) \subseteq \{b\}$.

    Due to Observation~\ref{obs:remove-cycles}, we may assume w.l.o.g.\ that the graph~$D$ underlying our financial market $M=(N,L,e,\alpha,\beta)$ is such that $G[\widetilde{N}]$ is acyclic: 
    Otherwise, we can repeat the following step as long as possible: if $C$ is a compression cycle containing only banks in~$\widetilde{N}$, then we replace~$M$ with $M-C$. By Observation~\ref{obs:remove-cycles}, there is a clearing vector~$p$ for~$M$ under which no bank in~$\widetilde{N}$ defaults if and only if there is a clearing vector for~$M-C$ under which no bank in~$\widetilde{N}$ defaults.
    Hence, $M$ admits a suitable compression if and only if $M-C$ admits one.

    \smallskip
    We proceed by computing for each $i \in \widetilde{N}$ the value $q_i=\max\{0,\sum_{j \in N} L_{ij}-e(i)-\sum_{j \in \widetilde{N}}L_{ji}\}$ which can be thought of as the payment that $i$ needs to receive from~$b$ in order to remain solvent (in the market without  compression). Note that if $q_i > L_{bi}$, then $i$ defaults even if $b$ pays all its liabilities towards~$i$, contradicting our assumption on~$b$. Hence, we may assume that $q_i\leq L_{bi}$ for each~$i \in \widetilde{N}$.
    
    We construct a flow network~$D'$ from~$D$ as follows: we first 
    set the capacity of each arc~$(i,j)$ as~$L_{ij}$. Then we split $b$ into~$b^-$ and~$b^+$ in the usual way, i.e., by replacing each arc $(b,i)$ with $(b^+,i)$ and each arc~$(j,b)$ with~$(j,b^-)$. We set~$b^+$ as the source and $b^-$ as the sink of the network. Notice that each compression for~$M$ corresponds to a flow in~$D'$ and vice versa; we write $M^f$ for the financial market obtained by applying the compression determined by some flow~$f$ in~$D'$.
    
    Let $f$ be a flow in~$D'$, let $|f|$ denote its \emph{size} (i.e., the total flow value leaving~$b^+$),  and let $p$ be a proportional clearing vector for~$M^f$. 
    The sum of liabilities in~$M^f$ entering~$b$ is  $\sum_{i \in \widetilde{N}}L_{ib}-|f|$,
    and since all banks other than~$b$ remain solvent, this equals the total incoming payment of~$b$ in under~$p$. Thus, the total payment leaving~$b$ under~$p$ is
    \[
    p_b=\alpha_b e(b)+\beta_b(\sum_{i \in \widetilde{N}}L_{ib}-|f|)=A_b-\beta_b |f|\]  
    where we let $A_b=\alpha_b e(b)+\beta_b \sum_{i \in \widetilde{N}}L_{ib}$. 
    Hence, for each $i \in \widetilde{N}$ we get that $i$ is solvent in under the proportional clearing vector for $M^f$ if and only if 
        \begin{equation}
        \label{eq:i-remains-solvent}
    \frac{L_{bi}-f_{bi}}{L_b -|f|}
    \cdot (A_b-\beta_b |f|) \geq q_i -f_{bi}.  
        \end{equation}

    If $\beta_b=1$, then the constraints~(\ref{eq:i-remains-solvent}) are, in fact, linear constraints on the variables $\{f_{bi}\}_{i \in \widetilde{N}}$. Hence, adding these constraints to the standard LP  describing a flow in~$D'$ and solving the obtained LP, we can find in polynomial time a flow~$f$ that satisfies~(\ref{eq:i-remains-solvent}) for each $i \in \widetilde{N}$ whenever such a flow exists.

\smallskip
    Hence, from now on we assume $\beta_b<1$.
    Let us define 
    \[\lambda_{|f|}=\frac{A_b-\beta_b |f|}{L_b-|f|}=\frac{A_b-\beta_b L_b}{L_b-|f|}+\beta_b,\]
    then condition (\ref{eq:i-remains-solvent}) can be rephrased as
    \begin{equation}
    \label{eq:lower-bound-convex_combination}
    \phi_i(f):=(1-\lambda_{|f|}) \cdot f_{bi} + \lambda_{|f|} \cdot L_{bi}
    \geq q_i .        
    \end{equation}

    Note that $|f| \leq L_b$ due to the definition of the network~$D'$. If $|f|=L_b$, then $b$ remains solvent after applying the compression corresponding to~$f$, which yields that all banks remain solvent, a contradiction to our assumption. Moreover, 
    if $p_b=A_b-\beta_b|f|\geq L_b-|f|$, then the total payment leaving~$b$ under~$p$ is at least the total liabilities of~$b$ in~$M^f$, i.e., $b$ remains solvent; again a contradiction. Therefore, we may assume that   $0 \leq p_b<L_b-|f|$ and so $0 \leq \lambda_{|f|} <1$. 

\begin{claim}
\label{clm:flowsize}
    Suppose $\beta_b<1$. If there exists a flow that satisfies~(\ref{eq:lower-bound-convex_combination}) 
    for each $i \in \widetilde{N}$, then there exists such a flow~$f$ that, in addition, either is a maximum-size flow in~$D'$, or has size $\frac{(\sum_{i \in \widetilde{N}}q_i)-A_b}{1-\beta_b}$.
\end{claim}

\begin{claimproof} 
    We distinguish between two cases.
    
    \myparagraph{Case A:}  $A_b-\beta_b L_b \geq 0$. In this case, $\odv{\lambda_{|f|}}{|f|}=\frac{A_b-\beta_b L_b}{(L_b-|f|)^2}\geq 0$, so $\lambda_{|f|}$ weakly increases if $|f|$ increases.
    
    Let $f^{\max}$ be a maximum flow in~$D'$. Note that we may assume that $f^{\max}_{bi}\geq f_{bi}$ for each $i \in \widetilde{N}$, as a maximum flow obtained from~$f$ using only shortest augmenting paths never decreases the flow value on any arc leaving~$b^+$.
    
    Observe further that $f_{bi} \leq L_{bi}$ by our construction of the network~$D'$, 
    and 
    $\phi_i(f)$  as defined by (\ref{eq:lower-bound-as-convex_combination})
    is the convex combination of $f_{bi}$ and~$L_{bi}$ with coefficients $1-\lambda_{|f|}$ and $\lambda_{|f|}$, respectively. Hence, taking~$f^{\max}$ instead of~$f$, 
    not only does the coefficient for the larger value ($L_{bi}$) weakly increase (since $\lambda_{|f^{\max}|} \geq \lambda_{|f|}$), 
    but additionally, the smaller value ($f_{bi}$) also weakly increases (since $f^{\max}_{bi}\geq f_{bi}$). Thus, we get that 
    \[(1-\lambda_{|f^{\max}|}) \cdot f^{\max}_{bi} + \lambda_{|f^{\max}|} \cdot L_{bi}
    \geq q_i \]
    holds for each $i \in \widetilde{N}$, which means that the compression corresponding to~$f^{\max}$ satisfies~(\ref{eq:lower-bound-as-convex_combination}) for each $i \in \widetilde{N}$, as required.

    \myparagraph{Case B:} $A_b-\beta_b L_b < 0$. 
     In this case, $\odv{\lambda_{|f|}}{|f|}=\frac{A_b-\beta_b L_b}{(L_b-|f|)^2} < 0$, meaning that $\lambda_{|f|}$  increases if $|f|$ decreases. Let us now examine how $\phi_i(f)$ changes if we decrease the flow value~$f_{bi}$ on some arc~$(b^+,i)$:
     \begin{align*}
     & \pdv{\left((1-\lambda_{|f|})  \cdot f_{bi} + \lambda_{|f|} \cdot L_{bi}\right)}{f_{bi}} =
     -\pdv{\lambda_{|f|}}{f_{bi}} \cdot {f_{bi}} + 
     (1-\lambda_{|f|}) +  \pdv{\lambda_{|f|}}{f_{bi}} \cdot L_{bi} \\
     &= 1-\beta_b-\frac{A_b - \beta_b L_b}{L_b-|f|}+(L_{bi}-f_{bi})\frac{A_b- \beta_b L_b}{(L_b-|f|)^2} \\
     &= 1 - \beta_b -\frac{A_b - \beta_b L_b}{(L_b-|f|)^2} \left( L_b -|f|-L_{bi}+f_{bi}\right) \\
     &= 1 - \beta_b -\frac{A_b - \beta_b L_b}{(L_b-|f|)^2} \left( \sum_{j \in  \widetilde{N}\setminus \{ i\}} (L_{bj} -f_{bj}) \right)  \geq 1-\beta_b > 0
     \end{align*}
     
     where we used that $L_{bj} -f_{bj} \geq 0$ for each $j \in \widetilde{N}$ and our assumption that $A_b-\beta_b L_b< 0$.
     %
    Therefore,
    decreasing the flow value~$f_{bi}$ but leaving $f_{bj}$ unchanged for all $j \in \widetilde{N} \setminus \{i\} $ strictly decreases $\phi_i(f)$ but increases $\phi_j(f)$ for each $j \in \widetilde{N} \setminus \{i\}$ due to the increase in~$\lambda_{|f|}$. 

    Let $f'$ be a flow in~$D'$ that satisfies~(\ref{eq:lower-bound-as-convex_combination}) for each $i \in \widetilde{N}$ and has minimum total size. 
   We claim that $\phi_i(f')=q_i$ for each~$i \in \widetilde{N}$.
   Otherwise, $\phi_i(f')>q_i$ for some~$i \in \widetilde{N}$, so by the above paragraph we can decrease $f'_{bi}$ slightly while maintaining the inequalities~(\ref{eq:lower-bound-as-convex_combination}) for each $j \in \widetilde{N}$. 
   However, this contradicts our choice of~$f'$, proving that 
   $f$ must fulfill (\ref{eq:lower-bound-as-convex_combination}) for each $i \in \widetilde{N}$ with equality, which also means that (\ref{eq:i-remains-solvent}) is satisfied with equality for each $i \in \widetilde{N}$.
    Summing up 
    these equalities, we get 
    \[A_b -\beta_b |f| = \left(\sum_{i \in \widetilde{N}} q_i \right)-|f|.\] 
    Due to $\beta_b<1$, this leads to $|f|=\frac{(\sum_{i\in \widetilde{N}} q_i) -A_b}{1-\beta_b}$, as required.
\end{claimproof}    
\smallskip

    Due to Claim~\ref{clm:flowsize}, we may assume that we know~$|f|$ and hence also $\lambda_{|f|}$.\footnote{Notice that by the proof of Claim~\ref{clm:flowsize}, we can completely determine the size of~$f$ depending on whether $A_b -\beta_b L_b\geq 0$.} 
    Then (\ref{eq:lower-bound-as-convex_combination}) 
    can be re-written as a simple lower bound 
    \begin{equation}
    \label{eq:lower-bounds-simpler}
        f_{bi} \geq \frac{q_i-L_{bi}\cdot \lambda_{|f|}}{1-\lambda_{|f|}}.
    \end{equation} 
    Using standard flow techniques, we can decide in polynomial time whether there exists a flow~$f$ achieving the required size~$|f|$ that, additionally, also respects the lower bound~(\ref{eq:lower-bounds-simpler}) for each $i \in \widetilde{N}$. If such a flow~$f$ exists, then each bank in~$\widetilde{N}$ is solvent under the proportional clearing vector for~$M^f$. If no such flow exists, then we conclude that there is no suitable compression for~$M$. 
    \end{proof}

\subsection{Priority-proportional cleaning model}    
\label{sec:algo:save-all-but-one:priority-proportional}

For the proof of the following generalization of Theorem~\ref{thm:allbutone}, see Appendix~\ref{app:prf-thmpolyallbutoneprio}.
\begin{restatable}[\linkproof{app:prf-thmpolyallbutoneprio}]{theorem}{thmpolyallbutoneprio}
\label{thm:poly:all-but-one:prio}
Given a financial market~$M$, one can decide in polynomial time whether $M$ admits a compression~$C$ such that $M-C$ admits a priority-proportional clearing vector~$p$ with $|\dft^{M-C}(p)|\leq 1$.
\end{restatable}

\section{A  Mixed Integer Linear Program for Optimal Compression}
\label{sec:ilp}

We present a MILP for finding the optimal priority proportional compression and clearing vector, given nonnegative endowements. 
Let $A=\{ (i,j) \in N\times N\colon L_{ij}>0\}$ be the set of arcs in the liability graph underlying the market.
For a bank~$i$, let $\delta (i)$ denote the set of outgoing arcs (debts) and $\rho (i)$ the set of incoming arcs (liabilities owed).

We assume the compression on arc $a$ is a discrete value $C_a = \sum_{\ell=0}^{\log L_a}2^{\ell}z_{a,\ell}$ where $z_{a,\ell}\in \{ 0,1\}$. 

We have proportion rate variables $\lambda_i^r \in [0,1]$ for each  bank $i\in N$ and priority class $g_i^r$, $r\in [k_i]$. Let $\delta_r(i)$ denote the set of arcs in priority class $g_i^r$.  Rate values must satisfy $1 \ge \lambda_i^1 \ge \dots \ge \lambda_i^{k_i} \ge 0$, and at most one $\lambda_i^r$ can be fractional (the transition point). 

\myparagraph{Variables:}
\begin{itemize}
    \item $q_i \in \{0,1\}$ for each $i \in N$: Indicator variable; $q_i=1$ if bank $i$ defaults.
    \item $z_{a,\ell} \in \{0,1\}$ for each $a \in A$, $\ell \in [\log L_a]_0$: Binary expansion variable for compression on arc~$a$.
    \item $\lambda_i^r \in [0,1]$ for each $i \in N$, $r \in [k_i]$: Proportion rate for priority class $r$ of bank $i$.
    \item $\mu_i^r \in \{0,1\}$  for each $i \in N$, $r \in [k_i]$: Auxiliary binary variable to enforce the structure of $\lambda_i^r$.
    \item $p_a \ge 0$ for each $a \in A$: Actual payment on arc $a$.
    \item $y_{a,\ell} \in [0,1]$ for each $a \in A$, $\ell \in [\log L_a]_0$: Variable for linearizing the product $\lambda_i^r \cdot z_{a,\ell}$ where $a \in \delta_r(i)$.
    \item $x_i \ge 0$ for each $i \in N$: Auxiliary variable for the total incoming payment $q_i \sum_{a \in \rho(i)} p_a$ of a defaulting bank.
\end{itemize}

\myparagraph{The Formulation:}

Let $R_i := \max \{e_i,0\} + \sum_{a \in \rho(i)} L_a$ be a trivial upper bound on bank $i$'s assets. 

\allowdisplaybreaks

\begin{align}
    \text{min} \quad & \sum_{i\in N} q_i 
    \notag 
    \\
    \text{s.t.} \quad 
    & \sum_{a\in \rho (i)} \sum_{\ell=0}^{\log L_a}2^{\ell}z_{a,\ell} = \sum_{a\in \delta (i)} \sum_{\ell=0}^{\log L_a}2^{\ell}z_{a,\ell}, 
    \quad \forall i \in N \tag{Circulation} \\
    & p_a = \lambda_i^r L_a - \sum_{\ell=0}^{\log L_a}2^{\ell}y_{a,\ell}, 
    \quad \forall i \in N, r \in [k_i], a \in \delta_r(i) \tag{Pay-Def} \\
    & y_{a,\ell} \le z_{a,\ell}, \quad y_{a,\ell} \le \lambda_i^r, \quad y_{a,\ell} \ge \lambda_i^r - (1 - z_{a,\ell}), \quad y_{a,\ell} \ge 0 \notag \\
    & \qquad \qquad \forall i \in N, r \in [k_i], a \in \delta_r(i), \ell \in [\log L_a]_0 \tag{McCormick} \\
        & \sum_{r=1}^{k_i} \mu_i^r = 1, \quad \forall i \in N \tag{Priority-Group} \\
    & \sum_{j=r+1}^{k_i} \mu_i^j \le \lambda_i^r \le \sum_{j=r}^{k_i} \mu_i^j, \quad \forall i \in N, r \in [k_i] \tag{Priority-Shape} \\
    & \sum_{a\in \delta (i)} \left(L_a - \sum_{\ell=0}^{\log L_a}2^{\ell}z_{a,\ell}\right) \le L_i q_i + \sum_{a\in \delta (i)} p_a, 
    \quad \forall i \in N \tag{Full-Pay} \\
    & (1 - (1-\alpha_i)q_i) e_i + \sum_{a\in \rho(i)} p_a - (1-\beta_i)x_i \ge \sum_{a\in \delta(i)} p_a, 
    \quad \forall i \in N \tag{Budget-LB} \\
    & (1 - (1-\alpha_i)q_i) e_i + \sum_{a\in \rho(i)} p_a - (1-\beta_i)x_i \le \sum_{a\in \delta(i)} p_a + (1-q_i)R_i, 
    \quad \forall i \in N \tag{Budget-UB} \\
    & x_i \le R_i q_i, \quad x_i \le \sum_{a\in \rho(i)} p_a, \quad x_i \ge \sum_{e\in \rho(i)} p_a - R_i(1-q_i), \quad x_i \ge 0 
    \quad \forall i \in N \tag{Def-Income}
\end{align}

\subsection*{Explanation and Proof of Correctness}

We now show that an optimal solution to the above MILP corresponds exactly to an optimal compression and a corresponding clearing vector adhering to priority-proportional payment rules.

\myparagraph{1. Compression Validity}
The constraint (Circulation) ensures that for every node, the total compression on incoming arcs equals the total compression on outgoing arcs. This is the definition of a valid compression circulation.

\myparagraph{2. Within Class Proportionality}
For an arc $a \in \delta_r(i)$, let $C_a = \sum 2^\ell z_{a,\ell}$ be the compression amount.
Constraints (McCormick) enforce $y_{a,\ell} = \lambda_i^r z_{a,\ell}$ because $z_{a,\ell}$ is binary.
Substituting this into (Pay-Def), we get
\[
p_a = \lambda_i^r L_a - \lambda_i^r C_a = \lambda_i^r (L_a - C_a).
\]
Thus, $p_a$ is exactly the fraction $\lambda_i^r$ of the remaining debt on arc $a$ for any $a\in \delta_r(i)$.

\myparagraph{3. Priority-Proportionality}
Constraints (Priority-Group) and (Priority-Shape) enforce the structure of the proportion rates.
By (Priority-Group), exactly one $\mu_i^p$ has valu~$1$.
Substituting into (Priority-Shape), we get
\begin{itemize}
    \item For $r < p$: $1 \le \lambda_i^r \le 1 \implies \lambda_i^r = 1$. (Highest priority groups paid in full).
    \item For $r = p$: $0 \le \lambda_i^r \le 1$. (Critical group paid partially).
    \item For $r > p$: $0 \le \lambda_i^r \le 0 \implies \lambda_i^r = 0$. (Lowest priority groups paid nothing).
\end{itemize}
This exactly models the priority-proportional rule.

\myparagraph{4. Default Logic and Budget Clearing}
We examine the consistency of the default variable $q_i$ and the budget constraints.

\textbf{Case A: Bank $i$ is solvent ($q_i = 0$).}
\begin{itemize}
    \item (Def-Income) forces $x_i = 0$.
    \item (Full-Pay) becomes $\sum (L_a - C_a) \le \sum p_a$. Since $p_a \le L_a - C_a$ by definition (as $\lambda \le 1$), this forces equality: $p_a = L_a - C_a$. Consequently, $\lambda_i^r = 1$ for all $r$, and the bank pays all debts fully.
    \item (Budget-LB/UB) reduce to $e_i + \sum_{a\in \rho (i)} p_a \ge \sum_{a\in \delta (i)} p_a$ and $e_i + \sum_{a\in \rho (i)} p_a \le \sum_{a\in \delta (i)} p_a + R_i$. Since the bank is solvent, it has enough assets to cover payments, satisfying the inequality and $R_i$ is large enough.
\end{itemize}

\textbf{Case B: Bank $i$ defaults ($q_i = 1$).}
\begin{itemize}
    \item (Def-Income) forces $x_i = \sum_{a \in \rho(i)} p_a$ (total incoming payments).
    \item (Full-Pay) becomes loose (true due to $\sum (L_a - C_a) \le L_i$), allowing $\lambda_i^r < 1$.
    \item (Budget-LB) and (Budget-UB) collapse to a single equality:
    \[
    \alpha_i e_i + \beta_i\sum_{a\in \rho(i)} p_a = \sum_{a \in \delta(i)} p_a.
    \]
    This states that the total outgoing payments $\sum_{a\in \delta (i)} p_a$ must exactly equal the bank's available assets after default costs $\alpha_i$ and $\beta_i$ are applied.
\end{itemize}

\myparagraph{Conclusion:}
The IP constraints ensure that payments are consistent with the compression, and follow the priority-proportional rules. Minimizing $\sum q_i$ maximizes the number of solvent banks. \hfill $\square$

\section{Simulations}
\label{sec:sim}

In this section, we evaluate the performance of our proposed MILP algorithm for finding an optimal compression in two scenarios: for a real-world dataset and for a synthetically generated one. For simplicity, we assumed a proportional clearing model. Relying on
Algorithm~\ref{alg:priority_proportional_clearing} for finding a maximal proportional clearing, 
we compare three different strategies: 
(1) the ``baseline'' clearing without any compressions, 
(2) the ``greedy'' approach to find a compression presented in Proposition~\ref{prop:greedy-saves-one-bank}, and 
(3) the MILP presented in Section~\ref{sec:ilp}. 

\myparagraph{Our datasets.}
The simulations were conducted on markets whose size ranges from~10 to~500 banks. For each size, we run simulations for 10 independently chosen instances. We assumed default costs of $\alpha $ chosen uniformly randomly per instance from $[0.4,0.8]$ (ratio of usable endowment) and $\beta $ chosen uniformly randomly per instance from $[0.6,0.9] $ (ratio of usable incoming payments) to reflect significant liquidity friction during a crisis.

The first dataset used in this study is based on a real-world financial network derived from the Romanian netting system, which is managed by the Romanian Ministry of Economy. This system employs a novel algorithm to clear financial obligations between companies, managing a significant volume of inter-company debt. The network is constructed from approximately 1.3 million invoices, representing the documented liabilities between a large number of market participants. 
We used a snowball sampling method\footnote{That is, after picking a starting bank uniformly at random, the algorithm samples the next bank from those banks towards which the last picked bank has some positive liability.} to preserve the realistic ``hub-and-spoke'' topology often found in Over-the-Counter markets. 
Initial endowments were randomized since they were not originally provided in the dataset:
we generated the initial endowment value of each bank~$i$ using a uniform distribution from 
$[0,0.8L_i]$ where $L_i$ is the total outward liability of bank~$i$.

To complement the real-world invoice data, we introduced a second experimental scenario using synthetic networks generated via the Erdős–Rényi model $G(n,p)$. While the real-world dataset exhibits sparse networks, our synthetic markets were generated with a fixed edge probability of $p=0.2$. This setup creates denser, more homogeneous networks that are rich in cyclic structures. These synthetic instances allow us to stress-test the compression algorithms in an environment where cyclic interdependencies are frequent, maximizing the theoretical potential for portfolio compression. 
Liabilities were generated uniformly at random between $100$ and $1,000$ units, providing a controlled baseline to contrast with the heavy-tailed distribution of the real invoice data.
%
Appendix~\ref{app:more-sim} contains additional results for cases where liabilities and/or endowments were generated by a log-normal distribution (instead of a uniform one).

\myparagraph{Results on real-world data.}
Our simulations showed that even with a snowball sampling, the resulting networks were very sparse, with an average degree of around 2 in most cases. This also suggests that the topology in this market is hierarchical in local clusters, being reflected in small and medium-sized supply chains that form open chains of credit (e.g., supplier → manufacturer → retailer) rather than closed loops.
The results for the Romanian netting dataset reveal that the greedy approach struggles to decrease the number of defaulting banks in this networks (it can only attain a decrease of at most~$1$), while the MILP can be much more effective. We list two possible reasons for the significant advantage of the MILP. First, the greedy approach finds liability cycles and applies compression, but only the MILP finds the optimal \emph{compression values} for the cycles found---the importance of this is illustrated via the situation of bank~$b$ in Example~\ref{ex}. Second, if cycles have common liability arcs, then greedy may only choose one of them, which then renders the other cycles infeasible later, while the MILP can consider all possible linear combinations of such cycles to find the most advantageous combination.


\begin{figure}[htbp]
    \centering
    \begin{subfigure}{0.49\textwidth}
\includegraphics[width=\textwidth]{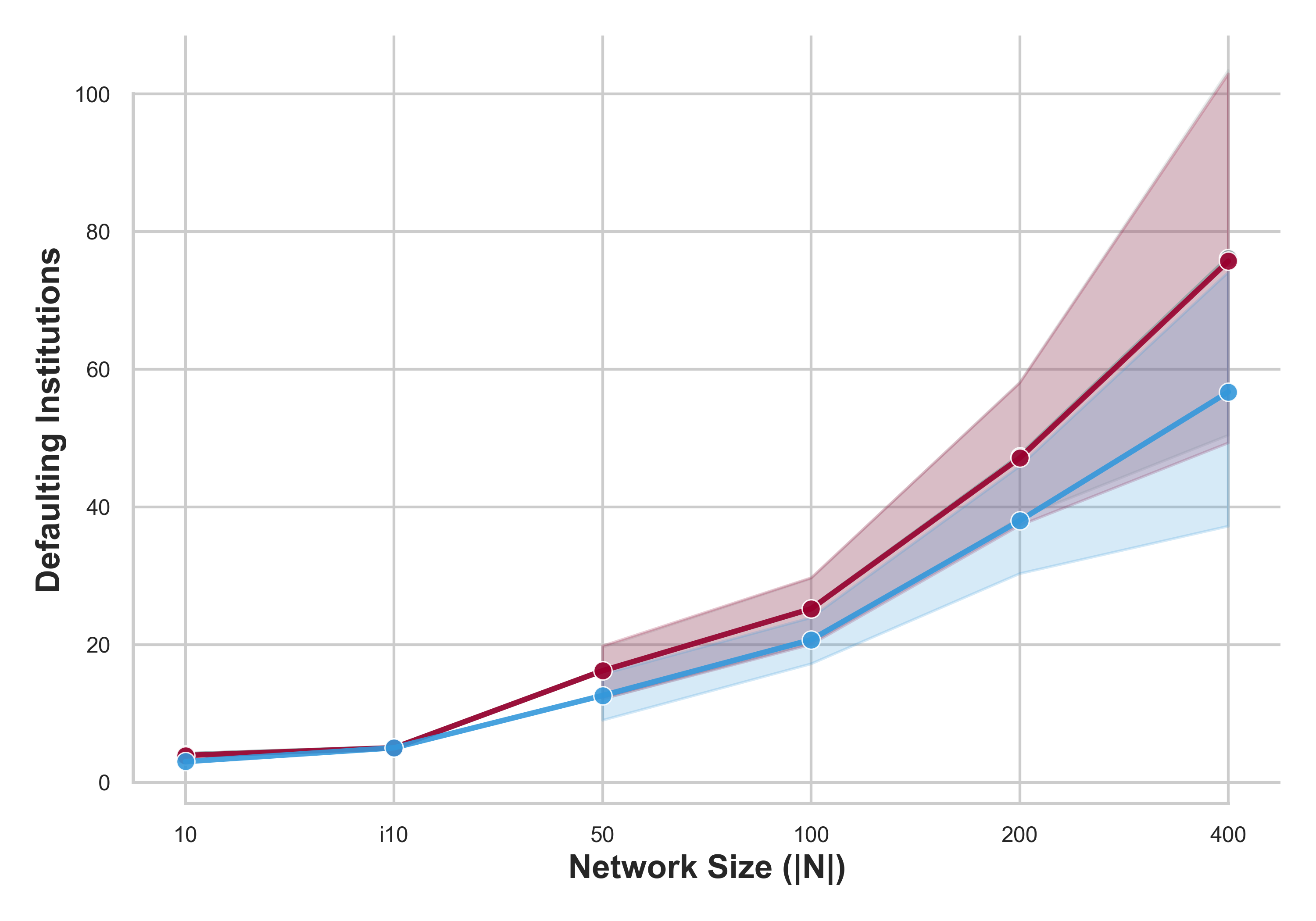}
    \subcaption{Results on the real-world dataset.}
    \label{fig:defaults}    
    \end{subfigure}
    \hfill
    \begin{subfigure}{0.49\textwidth}
    \includegraphics[width=\textwidth]{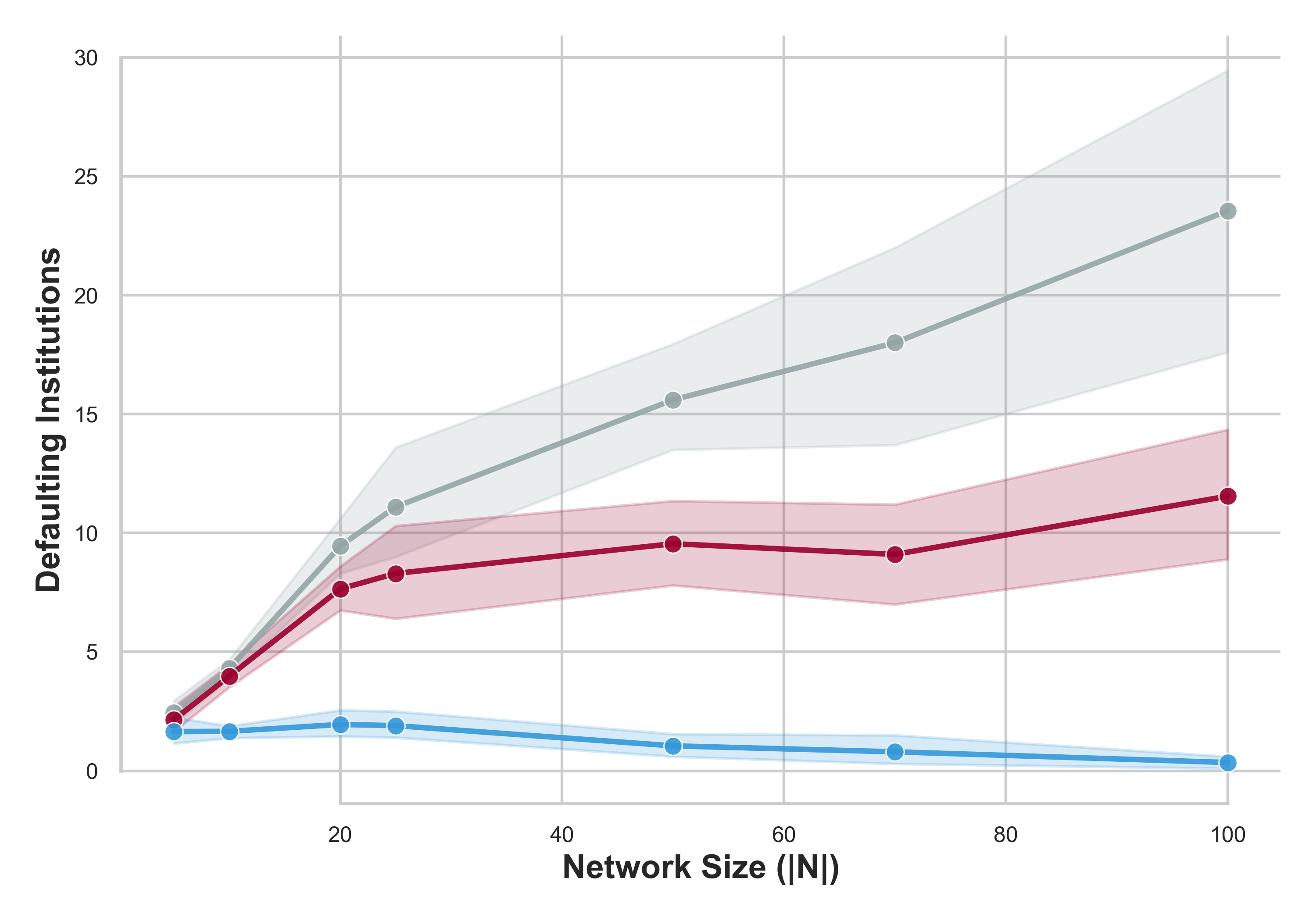}
    \subcaption{Results on the synthetic dataset.}
    \label{fig:milp}    
    \end{subfigure}
    \caption{Visual representation of our results. The plots show the number of defaulting banks for various sample sizes. Results for the \legendsquare{mygrey}~baseline, \legendsquare{myred}~greedy, and \legendsquare{myblue}~MILP methods are shown in grey, red, and blue, respectively. In Figure~(a), the red and grey lines overlap almost completely, showing the almost identical performance of the baseline and the greedy methods. The colored ``zones'' represent the $95\%$ confidence intervals.
    }
\end{figure} 

\myparagraph{Results on synthetic data.}
In high contrast to the results from the Romanian netting dataset, the results from the Erdős–Rényi model show that the systemic risk profile and algorithmic performance shift dramatically as cyclic interdependencies become the primary drivers of instability. The baseline algorithm---maximal proportional clearing without compression---resulted in the highest number of defaults in these dense topologies. This 
underscores the intrinsic fragility of a cyclic system without debt netting. Now that there are many cycles and liabilities compared to the Romanian dataset, even the greedy method effectively mitigates this by unlocking liquidity through cycle removal, though it remains substantially suboptimal compared to the MILP. Figure \ref{fig:defaults} compares the performance of these three methods in terms of the number of defaulting banks. 
Unlike the greedy method, which cancels every cycle it identifies, the MILP evaluates the systemic impact of each compression. By finding the optimal compression, it can choose to leave a cycle intact if the resulting payment flows are more beneficial for the solvency of a specific critical node. 
We even observe that even though the number of defaulting banks increases when the instance size gets larger, the number of defaulting banks
in the MILP solution gets even smaller. One interpretation of this phenomenon could be that the defaults do not increase as fast as the instance size, so the more solvent banks can help to save more of them in an optimal compression, which the MILP finds.
That being said, the MILP has its own downsides. It struggled to perform on iterations with more than~$100$ banks, by reaching the $1$-hour threshold per simulation.

\myparagraph{Conclusion.} Our MILP markedly outperforms the greedy method for both real-world and synthetic data in terms of the number of defaults, thanks to its ability to explore and optimize over an often exponentially large set of possible compression cycles allowing arbitrary compression values.

\section{Summary and Future Directions}
\label{sec:conclusions}

We have investigated the problem of finding a portfolio compression in financial markets in a way that minimizes the number of defaulting banks. We showed that this problem is $\NP$-hard; in fact, even deciding whether a single fixed bank can be saved from defaulting is $\NP$-hard. Despite the strong computational intractability of the problem, we were able to discover some positive results. First, we presented the first polynomial-time algorithm for simply computing a priority-proportional clearing in a financial market. Second, we gave a polynomial-time algorithm for determining whether the market admits a compression whose application results in at most one defaulting banks. 
Third, we gave a MILP formulation for our problem. We have demonstrated the usefulness of our MILP via extensive simulations on both synthetic and real-world data. 

Several open questions remain. 
Can our problem be solved efficiently if the underlying network of liabilities is \emph{almost} acyclic in some specific sense?
We remark that not even the case when there is only a single directed cycle in the network seems trivial.
Another intriguing question is whether the algorithm in Theorem~\ref{thm:poly:all-but-one:prio} can be extended for finding compressions where we aim to ensure that at most two banks default? What about larger constant values?

From a more general viewpoint, is it possible to extend our algorithmic results to more general models? An important generalization would be to allow for credit default swaps where the liability of a bank~$i$ towards some bank~$j$ depends also on whether some third bank~$k$ defaults or not. 






    \clearpage
    
    \bibliographystyle{ACM-Reference-Format}
    \bibliography{bib}

\begin{appendices}
\clearpage
\section{Omitted Proofs from Section~\ref{sec:alg-clearing}}

\subsection{Proof of Lemma~\ref{lem:sol}}
\label{app:prf-lemsol}
\lemsol*
\begin{proof}
     As in $(p^\star,z^\star)$, we have that $\EE_i(p^\star,z^\star)\ge 0$ for all $i\in N$, it is easy to see that it gives a solution to \ref{prob:min_subsidy}.
    Suppose that $(p',z')$ is another solution such that $\sum_{i\in N}z'_i<\sum_{i\in N}z_i^\star$. 
    
        Let $\delp_{ij}= p_{ij}'-p_{ij}^\star$ and $\delz_i = z_i'-z_i^\star$. 
             Define a set $S\subseteq N$, such that $i\in S$ if and only if $z_i^\star>0$ or ($\EE_i(p^\star,z^\star)=0$ and $\delp_{ij}>0$ for some $j\in \Gamma_i$). By the proportionality constraints, $\delp_{ij}>0$ for some $j\in \Gamma_i$ implies that $\delp_{ij}>0$ for all $j\in \Gamma_i$. Furthermore, if $z_i^\star>0$, then $\delp_{ij}\ge 0$  for all $j\in \Gamma_i$ as $p_{ij}^\star=0$ by~\specprop.
             This means that for all banks $i \in S$, we know $\delp_{ij}>0$ for all $j\in \Gamma_i$, which in turn implies $\delp_{ij}\geq 0$ for all $j \in N$.
    
        By property~\specprop, we have that for all $i\in S$, 
        the difference in net equity for bank~$i$ is $\delz_i + \sum_{j\in N} (\bet_{i}\delp_{ji} - \delp_{ij})\ge 0$,
        because $\EE_i(p^\star,z^\star)=0$
        and 
        $(p',z')$ satisfies (Budget) in \ref{prob:min_subsidy}.
%
    %
        Then, we have that $\sum_{i\in S}(\delz_i + \sum_{j\in N} (\bet_{i}\delp_{ji} - \delp_{ij}))\ge 0$. Since for $i,j\in S$, $\delp_{ij}\ge 0$ and $\bet_i\le 1$ we obtain
        $$ \sum\limits_{i\in S}\sum\limits_{j\in S}(\bet_i\delp_{ji}-\delp_{ij})=\sum\limits_{i,j\in S} (\bet_j-1)\delp_{ij} \le 0.$$    Therefore,    
        
        \begin{equation}
\label{eq:nonnegdiff}
0\le \sum_{i\in S}\Bigl(
\delz_i
+ \sum_{j\in S} (\bet_i\delp_{ji} - \delp_{ij})
+ \sum_{j\notin S} (\bet_{i}\delp_{ji} - \delp_{ij})
\Bigr)
\le 
\sum_{i\in S}\Bigl(
\delz_i
+ \sum_{j\notin S} (\bet_{i}\delp_{ji} - \delp_{ij})
\Bigr).
\end{equation} 

    For $i,j$ such that $i\notin S$, $j\in S$, we have that $\delp_{ij}\le 0$. Indeed, if $\EE_i(p^\star,z^\star)=0$, then $\delp_{ij}>0$ implies $i\in S$, contradiction; and if $\EE_i(p^\star,z^\star)>0$, then we must have that $p_{ij}^\star=L_{ij}$, so $\delp_{ij}\le 0$ by $p'_{ij}\le L_{ij}$. For $i\in S, j \notin S$ we have that $\delp_{ij}\ge 0$. Hence, $$\sum_{i\in S,j\notin S} (-\bet_{i}\delp_{ji} + \delp_{ij}))\ge 0.$$ Adding this to \eqref{eq:nonnegdiff}, we get that $\sum_{i\in S}\delz_i \ge 0$. 
    Since $\delz_i<0$ implies $z^\star_i>0$ from which $i \in S$ follows, we know that     
    $\delz_i\ge 0$ for each $i\notin S$. This implies $\sum_{i\in N}\delz_i\ge 0$, a contradiction. 
\end{proof}

\subsection{Proof of Lemma~\ref{lem:uniq}}
\label{app:prf-lemuniq}
\lemuniq*
\begin{proof}
     Let $(p,z)$ and $(p',z')$ be two optimal solutions of \ref{prob:min_subsidy} that satisfy property~\specprop. Create a vector $(p'',z'')$ as follows. Let $p''=\max \{ p,p'\}$,  and $z''=\min\{ z,z'\}$, where max and min are taken coordinate-wise. We claim that $(p'',z'')$ is a solution of \ref{prob:min_subsidy}, which implies that it must hold that $z_i=z_i'$ for all $i\in N$, as otherwise $\sum_{i\in N}z_i''<\sum_{i\in N}z_i$, contradicting the optimality of $(p,z)$.
    
    If $z_i=z_i'$, then it is clear that $\EE_i(p'',z'')\ge 0$ holds, as by the priority-proportionality constraint (Prop), either $p_{ij}\ge p_{ij}'$ for all $j\in N$ or $p_{ij}\le p_{ij}'$ for all $j\in N$ (that is, the $\max \{ p_{ij},p'_{ij}\}$ is $p_{ij}$ for all $j \in N$ or it is $p'_{ij}$ for all $j \in N$). Therefore, taking the maximum on the incoming and outgoing arcs cannot destroy the property that the net equity is non-negative. By the same reason, the proportionality constraints are also satisfied and so are all other constraints  in \ref{prob:min_subsidy}. 
    
    Suppose that $z_i \neq z'_i$; by symmetry, we may assume $z_i<z_i'$. Then, $p_{ij}'=0$ by property~\specprop, so $p_{ij}''=\max \{ p_{ij},p_{ij}'\} = p_{ij}$ for all $j\in \Gamma_i$ (and thus for all $j\in N$). Hence, similarly as before, $\EE_i(p'',z'')\ge 0$ is satisfied too and so are all constraints in \ref{prob:min_subsidy}.

    Therefore, $(p'',z'')$ is indeed a solution to \ref{prob:min_subsidy}, implying $z_i=z_i'$ for $i\in N$.
\end{proof}

\subsection{Proof of Lemma~\ref{lem:minp}}
\label{app:prf-lemminp}
\lemminp*
\begin{proof}
    Let $(p,z)$ be optimal for \ref{prob:min_payment}. Since it satisfies the constraint $\sum z_i = Z^*$, it is optimal for \ref{prob:min_subsidy}. Suppose it does not satisfy~\specprop. Then there exists some $i$ with $z_i > 0$ and $p_{ij} > 0$ for some $j \in \Gamma_i$ (the case $\EE_i(p,z)>0$ is impossible: then, $z_i$ could be decreased by some small $\varepsilon>0$ without violating any constraint of \ref{prob:min_subsidy}, contradicting its optimality).

    Again, we have $p_{ij} > 0$ for all $j \in \Gamma_i$ by the priority-proportionality constraint (Prop).
    We construct a new solution. 
    Define $\varepsilon=\min\{x:x \leq z_i,x \frac{L_{ij}}{L_i^{\gamma_i}}\leq p_{ij} \text{ for all }j \in \Gamma_i\}$, and notice that $\varepsilon>0$ because $z_i>0$ and $p_{ij}>0$ for all $j \in \Gamma_i$. Now, let decrease the subsidy~$z_i$ by $\varepsilon$ and, for each $j \in \Gamma_i$, decrease the payment $p_{ij}$ by $x \frac{L_{ij}}{L_i^{\gamma_i}}$ and simultaneously increase the subsidy~$z_j$ by the same amount
    so that we maintain the (Budget) inequality for bank~$j$. Observe that this yields a solution to \ref{prob:min_payment}: the (Budget) inequality for~$i$ is maintained, and the payments remain priority-proportional.  
    Moreover, the total sum $\sum_{i\in N} z_i$ of the subsidies remains the same, but the total payment $\sum_{i,j\in N} p_{ij}$ strictly decreases. This contradicts the optimality of $(p,z)$ to \ref{prob:min_payment}.
    Thus, the solution must satisfy~\specprop.
\end{proof}

\subsection{Proof of Lemma~\ref{lem:zero}}
\label{app:prf-lemzero}
\lemzero*
\begin{proof}
It is clear that $(p^\star,z^\star)$ is a solution to \ref{prob:max_payment}, as it is an optimal solution to \ref{prob:min_subsidy} by Lemma~\ref{lem:sol} and it satisfies~\specprop. It remains to show that it is the unique optimal solution. Hence, to prove the lemma, it is enough to prove that for any solution $(p,z^\star)$ to \ref{prob:max_payment}, we have that $p_{ij}\le p_{ij}^\star$ for all $i,j\in N$.

Suppose for the contrary that there is a solution $(p,z^\star)$ and banks $i,j\in N$ such that $p_{ij}>p_{ij}^\star$.
If $z_i^\star>0$, then this would mean $p_{ij}=p^\star_{ij}$ by the constraints of \ref{prob:max_payment}, a contradiction. 

If $z_i^\star=0$, then $z_i^k=0$ for all $k\in \mathbb{N}$.
Then, by the indirect assumption, there exists some $k\in \mathbb{N}$, such that $p_{ij}^{k+1}<p_{ij}\le p_{ij}^{k}$, using that $p_{ij}^0=L_{ij}$ for $j\in G_i^{\gamma_i}$ and $p_{ij}^0=p_{ij}=0$ otherwise. Choose $i$ and $j$ such that $k$ is as small as possible. Since $p_{ij}^k\ne p_{ij}^{k+1}$, we have that $\sum_{j'\in N}p_{ij'}^{k+1}=\alp_ie_i + \sum_{j'\in N}\bet_ip_{j'i}^k$. Then, we get $\sum_{j'\in N}p_{ij'}\le \alp_ie_i + \sum_{j'\in N}\bet_ip_{j'i}\le \alp_ie_i+\sum_{j'\in N}\bet_ip_{j'i}^k=\sum_{j'\in N}p_{ij'}^{k+1}$ by the choice of $k$. By the priority-proportionality constraint of \ref{prob:min_subsidy}, $p_{ij'}\le p_{ij'}^{k+1}$ for all $j'\in N$, a contradiction.
\end{proof}

\section{Simplification for Proportional Clearing with Nonnegative Endowments}
\label{app:simplify}
We show that if the initial endowments $e_i$ are nonnegative, and we only need a proportional clearing for our market (instead of a priority-proportional one), then it suffices to solve the following single LP repeatedly:
\begin{equation}
\tag{\textsc{LP:Max-payment-prop}}
\label{prob:max_payment_prop}
\begin{array}{rlcll}
    \max & \multicolumn{4}{l}{\sum\limits_{i \in N} p_{ij}} \\[4pt]
    \text{s.t.}
    & \sum_{j\in N}p_{ij} & \leq & \alp_i e_i + \sum_{j\in N}\bet_{i}p_{ji} & \forall i \in N \quad \text{(Budget)} \\[2pt]
    & p_{ij} & = & \lambda_iL_{ij} & \forall i \in N \quad \text{(Prop)} \\[2pt]
    & \lambda_i& \in & [0,1] & \forall i\in N, 
\end{array}
\end{equation}

In Algorithm~\ref{alg:proportional_clearing}, we show how we can simplify the algorithm  of Theorem~\ref{thm:alg-pripropclear} if we only aim to compute a proportional clearing.

\begin{algorithm}[ht]
\caption{Proportional Clearing for Nonnegative Endowments}
\label{alg:proportional_clearing}

\ForEach{$i \in N$}{
    Set $\;(\alpha_i^0,\beta_i^0) \gets (1,1)$\;
}

\For{$r \gets 0$ \KwTo $\infty$}{
    $\mathrm{Changed \gets False}$\;
    Solve \ref{prob:max_payment_prop} to obtain $p^r$\;
    \ForEach{$i \in N$}{
        \If{$\alpha_i^r e_i +  \sum_{j\in N} \beta_i^r p_{ji}^r < L_i$
        \textbf{and} $(\alpha_i^r,\beta_i^r)=(1,1)$}{
            Set $\;(\alpha_i^{r+1},\beta_i^{r+1}) \gets (\alpha_i,\beta_i)$\;
            $\mathrm{Changed \gets True}$\;
        }
        \Else{ 
        Set $\;(\alpha_i^{r+1},\beta_i^{r+1}) \gets (\alpha_i^r,\beta_i^r)$\;}
    }
    \If{$\mathrm{Changed=False}$}{
    \Return $p^r$\;}
}
\end{algorithm}
\begin{theorem}
    Algorithm~\ref{alg:proportional_clearing} computes a coordinate-wise maximum proportional clearing vector if initial endowment values are nonnegative.
\end{theorem}
\begin{proof}
    Since all banks $i$ have a single priority group, if we have $z_i^r>0$ in Algorithm~\ref{alg:priority_proportional_clearing} then $\gamma_i^r$ will decrease to $0$ and we can conclude that $\alp_i^r e_i \le \alp_i^r e_i +\sum_{j\in N}\bet_i^rp_{ji}^r<0$, implying $e_i<0$ which contradicts our initial assumption for this section. Hence, we can assume that we need no subsidies at all. Thus, computing $p^r$ simplifies to solving \ref{prob:max_payment_prop}. Indeed, it has a solution if and only if \ref{prob:min_subsidy} has a solution with $z^r=0^N$, (setting $k_i=1 \; \forall i\in N$). Thus, we can assume $z^r=0^N$, so \ref{prob:max_payment} becomes \ref{prob:max_payment_prop}, since the extra constraints are never needed (they are only needed for $i$ with $z_i^r>0$), also implying that \ref{prob:max_payment_prop} is always feasible.
\end{proof}

\section{Omitted proofs from Section~\ref{sec:np-hardness}}

Before providing the proofs for each of our intractability results, let us first state the $\NP$-hard computational problems that we use in our reductions, namely, \maxsat\ and \partition.

\noindent
\begin{center}
\begin{minipage}{\textwidth}
\fbox{
    \begin{tabular}{@{\hspace{2pt}}l@{\hspace{4pt}}p{0.84\textwidth}@{\hspace{2pt}}}%
        \multicolumn{2}{@{\hspace{2pt}}p{0.87\textwidth}}{
        \maxsat
        } \\
        {\bf Input:} &
            A \twosat\ formula $\varphi$ with clauses $C_1,\dots,C_m$ and variables $X_1,\dots ,X_n$ with a number~$K$. 
        \\
        {\bf Question:} &
            Is there a truth assignment $f:[n]\to \{ T,F\}$, such that at least $K$ clauses are satisfied?
    \end{tabular}
    }
\end{minipage}
\end{center}
\medskip
\maxsat\ is known to be $\NP$-hard  even if each variable appears at most two times (but at least once) in positive and at most two times (but at least once) in negative form~\cite{berman2002some}.

The \partition\ problem is also a classic problem known to be $\NP$-hard~\cite{Karp1972}.
\noindent
\begin{center}
\begin{minipage}{\textwidth}
\fbox{
    \begin{tabular}{@{\hspace{2pt}}l@{\hspace{4pt}}p{0.84\textwidth}@{\hspace{2pt}}}%
        \multicolumn{2}{@{\hspace{2pt}}p{0.87\textwidth}}{
        \partition
        } \\
        {\bf Input:} &
            A set of integers $a_1,\dots, a_n$. 
        \\
        {\bf Question:} &
            Is there a subset $S\subseteq [n]$, such that $\sum_{i\in S}a_i = \sum_{i \in [n] \setminus S} a_i$?
    \end{tabular}
    }
\end{minipage}
\end{center}
\medskip

\subsection{Proof of Theorem \ref{thm:nph:findcomp}}
\label{app:prf-thmnphfindcomp}
\thmnphfindcomp*
\begin{proof}
We prove the result in two steps.

\medskip
\noindent
\textbf{Reduction for~(a) and~(b).}
    To prove the hardness for the two problems in~(a) and~(b), we provide a single reduction from \maxsat. Let $\varphi = C_1\wedge \dots \wedge C_m$ be a \textsc{2-sat} formula with variables $X_1,\dots, X_n$ and $K$ be the parameter.

    Let $Q=2m+6n+1$.
    We create a financial network with no default costs. For each variable $X_i$, we have banks $i,i'$ and banks $T_i^0,T_i^1,\dots ,T_i^Q$ as well as $F_i^0,F_i^1,\dots , F_i^Q$. For each clause $C_j$ we create two banks $c_j,c_j'$.
    
    For each $i\in [n]$, bank $i$ has an initial endowment of $e_i=3$. All other banks have initial endowment 0. The liabilities are as follows:
        \begin{itemize}
            \item $L_{i,i'}=7$ for $i\in [n]$.
            \item $L_{i,T_i^0}=L_{T_i^0,i}=L_{i,F_i^0}=L_{F_i^0,i}=16$ for $i\in [n]$.
            \item $L_{T_i^0,T_i^1}=L_{F_i^0,F_i^1}=2^{100}$.
            \item $L_{T_i^{l},T_i^{l+1}}=L_{F_i^{l},F_i^{l+1}}=2$ for $i\in [n],l\in [Q-1]$.
            \item $L_{T_i^Q,c_j}=1$ whenever $X_i\in C_j$, otherwise 0.
            \item $L_{F_i^Q,c_j}=1$ whenever $\overline{X}_i\in C_j$, otherwise 0.
            \item $L_{c_j,c_j'}=1$ for $j\in [m]$.
        \end{itemize}
    
    If $X_i$ (resp. $\overline{X}_i$) appears in only one clause then we set $L_{T_i^Q,c_j}=2$ (resp. $L_{F_i^Q,c_j}=2$) instead.
    
    We claim that there is a compression, such that at least $nQ+K+m+n$ banks remain solvent, if and only if there is a truth assignment that satisfies at least $K$ clauses. Furthermore, if there is such a compression, then it is a bilateral compression. The last statement follows from the fact that the only cycles in the digraph whose arcs are the nonzero liabilities are the cycles between $i$ and $T_i^0$ or $i$ and $F_i^0$ for $i\in [n]$.
    
    Suppose that it is possible that at least $nQ+K+m+n$ banks remain solvent in the maximal proportional clearing after a compression $C$. We first show that for any $i\in [n]$, either all banks $T_i^l$, $l\in Q$ or all banks $F_i^l$, $l\in [Q]$ must remain solvent. 

    Since $p_{ji}\le L_{ji}$ for any clearing and any $j$ and $e_i=3$, the total money outflow from $i$ towards $\{ T_i^0,F_i^0\}$ is bounded by $3\cdot\frac{32}{39} +16\cdot 2 \cdot\frac{16}{2^{100}+16}<2.5$, using proportionality and that $\frac{7}{32}\cdot 3$ must arrive at $i'$. Hence, the banks $T_i^1,F_i^1$ receive less than $2.5$ together. Hence, at most one of them can get enough to pay its liabilities fully and when one cannot, then neither can any bank in its chain ($T_i^1,\dots, T_i^Q$ or $F_i^1,\dots, F_i^Q$) until $F_i^{Q}$ or $T_i^{Q}$. Since $Q$ is larger than the number of banks not of the form $T_i^l,F_i^l$ for some $i\in [n],l\in [Q]$, it follows that saving $nQ+K+m+n$ banks is only possible, if all banks are solvent in one of the chains for each $i\in [n]$. 
    
    It is easy to see that the banks $i,T_i^0,F_i^0$ always default, while the banks $i',c_j'$ never default and there are $n+m$ of the latter. Hence, to have $nQ+K+m+n$ banks not defaulting, we must have at least $K$ such $c_j$ banks. 
    
    As $T_i^1,F_i^1$ get less than $2.5$ together, but one of them gets at least $2$, we get that the other gets strictly less than $0.5$. Hence, for a bank $c_j$ to remain solvent, at least one of its two debtor must be an agent $T_i^Q$ or $F_i^Q$ that has at least 2.
    
    Hence, we can define a truth assignment, such that $X_i$ is True, if $T_i^1$ is solvent and False, if $F_i^1$ is solvent and this will satisfy at least $K$ clauses.
    
    In the other direction, take a truth assignment that satisfies at least $K$ clauses. For each $i\in [n]$, if $X_i$ is true, then compress the cycle $\{(i,F_i^0),(F_i^0,i)\}$ completely, and otherwise compress the cycle $\{ (i,T_i^0),(T_i^0,i)\}$ completely. 
    
    Then, take an $i\in [n]$ and say $X_i$ was true. Then, $T_i^1$ will get at least $\frac{16}{23}\cdot 3\cdot \frac{2^{100}}{2^{100}+16}>2$, else $F_i^1$ will get at least $2$ in the maximal proportional clearing. It is straightforward to verify that each bank $c_j$ such that  $C_j$ has a true literal will remain solvent and so will all banks in one of the chains for each $i\in [n]$ as well as the banks $c_j'$ and $i'$. Hence, it gives a compression with at least $nQ+k+n+m$ banks solvent as desired. 

\medskip
\noindent
\textbf{Reduction for~(c).}
Next, we prove the hardness of 
    \FindCompCycle. 
      Again, we reduce from \maxsat. Let $\varphi = C_1\wedge \dots \wedge C_m$ be a \textsc{2-sat} formula with variables $X_1,\dots, X_n$ and $K$ be the parameter.

        Let $Q=2m+6n+1$.
        We create a financial network with no default costs as follows. For each variable $X_i$, we have banks $i,i'$ and banks $T_i^0,T_i^1,\dots ,T_i^Q$ as well as $F_i^0,F_i^1,\dots , F_i^Q$. For each clause $C_j$ we create two banks $c_j,c_j'$.
    
        For each $i\in [n]$, bank $i$ has an initial endowment of 3. All other banks have initial endowment 0. The liabilities are as follows:
        \begin{itemize}
            \item $L_{i,i'}=7$ for $i\in [n]$.
            \item $L_{i,T_i^0}=L_{T_i^0,i+1}=L_{i,F_i^0}=L_{F_i^0,i+1}=16$ for $i\in [n]$, where $n+1:=1$.
            \item $L_{T_i^0,T_i^1}=L_{F_i^0,F_i^1}=2^{100}$.
            \item $L_{T_i^{l},T_i^{l+1}}=L_{F_i^{l},F_i^{l+1}}=2$ for $i\in [n],l\in [Q-1]$.
            \item $L_{T_i^Q,c_j}=1$ whenever $X_i\in C_j$, otherwise 0.
            \item $L_{F_i^Q,c_j}=1$ whenever $\overline{X}_i\in C_j$, otherwise 0.
            \item $L_{c_j,c_j'}=1$ for $j\in [m]$.
        \end{itemize}
    
    If $X_i$ (resp. $\overline{X}_i$) appears in only one clause then we set $L_{T_i^Q,c_j}=2$ (resp. $L_{F_i^Q,c_j}=2$) instead.
    
    We claim that there is a compression cycle, such that at least $nQ+K+m+n$ banks remain solvent, if and only if there is a truth assignment that satisfies at least $K$ clauses.
    
    Suppose that it is possible to to have at least $nQ+K+m+n$ banks in the maximal proportional clearing remain solvent after a compression $C$ involving a single cycle. We fist show that for any $i\in [n]$, either all banks $T_i^l$, $l\in [Q]$ or all banks $F_i^l$, $l\in [Q]$ must be solvent.

     Since $p_{ij}\le L_{ij}$ for any clearing and $e_i=3$, the total money outflow from $i$ towards $\{ T_i^0,F_i^0\}$ is bounded by $3\cdot\frac{32}{39} +16\cdot 2 \cdot\frac{16}{2^{100}+16}<2.5$, using proportionality and that $\frac{7}{32}\cdot 3$ must arrive at $i'$. Hence, the banks $T_i^1,F_i^1$ receive less than $2.5$ together.
  
Since $Q$ is larger than the number of banks not of the form $T_i^l,F_i^l$ for some $i\in [n],l\in [Q]$, it follows that saving $nQ+K+m+n$ banks is only possible, if all banks remain solvent in one of the chains for each $i\in [n]$. 
    
    It is easy to see that the banks $i,T_i^0,F_i^0$ always default, while the banks $i',c_j'$ never default and there are $n+m$ of the latter. Hence, to have $nQ+K+m+n$ banks not defaulting, we must at least $K$ of the $c_j$ banks remaining solvent. 
    
    As $T_i^1,F_i^1$ get strictly less than $2.5$ together, but one of them gets at least $2$, we get that the other gets strictly less than $0.5$. Hence, for a bank $c_j$ to remain solvent, at least one of its two debtor must be an agent $T_i^Q$ or $F_i^Q$ that has at least 2.
    
    Hence, we can define a truth assignment, such that $X_i$ is True, if $T_i^1$ is solvent and False, if $F_i^1$ is solvent and this will satisfy at least $K$ clauses.
    
    In the other direction, take a truth assignment that satisfies at least $K$ clauses. We create a compression cycle $C$ as follows. For $i\in [n]$, if $X_i$ is true, then we add the arcs $\{(i,F_i^0),(F_i^0,i+1)\}$, and otherwise we add $\{(i,T_i^0),(T_i^0,i+1)\}$ and then compress the cycle completely, i.e., set $C_a=10$ for all $a\in C$.
    
    Then, take an $i\in [n]$ and say $X_i$ was true. Then, $T_i^1$ will get at least $\frac{16}{23}\cdot 3>2$, else $F_i^1$ will get at least $2$ in the maximal proportional clearing. It is straightforward to verify that each bank $c_j$ such that  $C_j$ has a true literal will remain solvent and so will all banks in one of the chains for each $i\in [n]$ as well as the banks $c_j'$ and $i'$. Hence, it gives a compression with at least $nQ+k+n+m$ banks solvent in $M-C$ as desired. 
\end{proof}

\subsection{Proof of Theorem~\ref{thm:nph:all-but-three}}
\label{app:prf-thmnphallbutthree}
\thmnphallbutthree*
\begin{proof}
We reduce from \partition. Let $I=\{ a_1,\dots, a_n\}$ be an instance of \partition. Clearly, we can assume that $\sum_{i\in [n]}a_i$ is even, else there cannot be any solution. We create an instance $I'$ of \FindComp\ and an instance $I''$ of \FindCompBank\ as follows.
For both, we have no default costs.

\begin{itemize}
    \item For each $i\in [n]$, we create banks $\hat{x}_i,x_i,y_i,x_i^*$. Here, $e_{x_i}=e_{y_i}=a_i$ and $e_{\hat{x}_i}=e_{x^*_i}=0$.
    \item We create banks $b_S,b_{-S},b',b^*$. We have $e_{b_S}=e_{b_{-S}}=e_{b'}=e_{b^*}=0$.
\end{itemize}

Let $R=4n\sum_{i\in [n]}a_i^2+1$.
The liabilities are as follows. 
\begin{itemize}
    \item $L_{x_i,x_i^*}=L_{y_i,x_i^*}=L{x_i^*,\hat{x}_i}=L_{\hat{x_i},y_i}=L_{\hat{x}_i,x_i}=R$ for $i\in [n]$.
    \item $L_{x_i^*,b^*}=4R^2n+1$ and $L_{x_i,b_S}=L_{y_i,b_{-S}}=a_i+\frac{1}{2n}$ for $i\in [n]$.
    \item $L_{b_S,b'}=L_{b_{-S},b'}=\frac{\sum_{i\in [n]}a_i}{2}$, and $L_{b',b^*}=\sum_{i\in [n]}a_i$.
\end{itemize}
For~$I'$, we set the maximum number of defaulting banks as $k=|N|-3$.
For $I''$, we let the fixed bank be $b'$.

\begin{claim}
\label{clm:NP1}
    If $I''$ is a yes-instance for \FindCompBank, then $I$ is a yes-instance for \partition.
\end{claim}
\begin{claimproof}
    Suppose that $b'$ does not default in $M-C$ after a compression $C$. 

First, we claim that for each $i\in [n]$ at least one of $p_{x_i,b_S}<\frac{1}{n}$ or $p_{y_i,b_{-S}}<\frac{1}{n}$ holds. For this, consider the bank $x_i^*$. It receives a payment of at most $2R$. Of this payment, at most $\frac{R}{R+4R^2n+1}$ fraction will be paid to $\hat{x}_i$ (with equality only if compression~$C$ has zero value on~$(x_i^*,\hat{x}_i)$). Hence, the payment that $\hat{x}_i$ receives is at most $\frac{1}{2n}$, irrespective of the compression $C$. 

Since any compression cycle that contains either of the arcs $(x_i,x_i^*),(y_i,x_i^*)$ also contains the arc~$(x_i^*,\hat{x}_i)$, we get that at least one of the liability values $L_{x_i,x^*_i}-C_{x_i,x^*_i}$ and $L_{y_i,x^*_i}-C_{y_i,x^*_i}$ obtained after applying compression~$C$ is at least $\frac{R}{2}$. 
Suppose by symmetry that $L_{x_i,x_i^*}-C_{x_i,x_i^*} \geq \frac{R}{2}$.

Note that $x_i$ and $y_i$ only receive payment from $\hat{x}_i$, so the total outgoing payment of~$x_i$ is at most $\frac{1}{2n}+a_i$.
Then, the payment $p_{x_i,b_S}$ is at most $(a_i+\frac{1}{2n})\cdot \frac{a_i}{a_i+R/2}<\frac{1}{2n}$ by the choice of $R$.

Hence, we obtain that for each $i\in [n]$, we can choose one $x_i$ or $y_i$ such that the total payment for all of these $n$ banks towards $\{ b_S,b_{-S}\}$ is strictly less than $\frac{1}{2}$. Let the set of these banks be $B$.

Since $b_S,b_{-S}$ and $b'$ cannot be included in any compression cycles, the only way that $b'$ does not default is if $p_{b_S,b'}=p_{b_{-S},b'}=\frac{\sum_{i\in [n]}a_i}{2}$.
Therefore, setting $S=\{ i\in [n]\mid y_i\in B\}$, we have that \[\sum_{i\in S}(a_i+\frac{1}{2n})>\frac{\sum_{i\in [n]}a_i}{2}-\frac{1}{2}\] and 
\[\sum_{i\in [n] \setminus S}(a_i+\frac{1}{2n})>\frac{\sum_{i\in [n]}a_i}{2}-\frac{1}{2}.\] Since each $a_i$ is integer, this implies that $\sum_{i\in S}a_i=\sum_{i\in [n] \setminus S}a_i=\frac{\sum_{i\in [n]}a_i}{2}$, as desired.
\end{claimproof}

\begin{claim}
    If $I'$ is a yes-instance for \FindComp, then $I''$ is a yes-instance for \FindCompBank.
\end{claim}
\begin{claimproof}
    If $I'$ is a yes-instance for \FindComp, then at least $3$ banks must not default in the market $M-C$ after some compression $C$. It is clear that $x_i^*,x_i,y_i$ and $\hat{x}_i$ default for all $i\in [n]$ as they have more outgoing liabilities than their total possible capital after collecting their payments from others: this is obvious for $x_i^*$ and~$\hat{x}_i$, for~$x_i$ and~$y_i$ follows from the fact that $\hat{x}_i$ receives at most~$\frac{1}{2n}$ incoming payment irrespective of the compression applied, as we have seen in the proof of Claim~\ref{clm:NP1}. Hence, for at least three banks to be solvent under some proportional clearing vector for~$M-C$, it must be the case that besides~$b^*$, at least two banks from~$\{ b_S,b_{-S},b'\}$ do not default. Since $b'$ only remains solvent if and only if both $b_S$ and $b_{-S}$ do too, this is equivalent to saying that $b'$ remains solvent, showing that $I''$ is a yes-instance.
\end{claimproof}

\begin{claim}
\label{clm:part-to-clearing}
    If $I$ is a yes-instance for \partition, then $I'$ is a yes-instance for \FindComp.
\end{claim}
\begin{claimproof}
    Let $S \subseteq [n]$ be a set of indices such that $\sum_{i\in S}a_i=\sum_{i\in [n] \setminus S}a_i$. We define a compression $C$ as follows. For each $i\in S$, we let $C_{x_i^*,\hat{x}_i}=C_{\hat{x}_i,x_i}=C_{x_i,x_i^*}=R$, and for each $i \in [n] \setminus S$, we let $C_{x_i^*,\hat{x}_i}=C_{\hat{x}_i,y_i}=C_{y_i,x_i^*}=R$.

Consider the market $M-C$. For some $i\in S$, $x_i$ has positive liability only towards~$b_S$. Hence, all of its endowment~$a_i$ is paid to $b_S$. Similarly, for some $i\in [n] \setminus S$, $y_i$ has positive liability only towards~$b_{-S}$, so all of its endowment $a_i$ is paid to $b_{-S}$. Therefore, both $b_S$ and $b_{-S}$ receive $\frac{\sum_{i\in [n]}a_i}{2}$, so they remain solvent and so does $b^*$. Hence, $I'$ is a yes-instance as claimed.
\end{claimproof}

    The statement of the theorem now follows from Claims~\ref{clm:NP1}--\ref{clm:part-to-clearing}. 
\end{proof}

\section{Proof of Theorem~\ref{thm:poly:all-but-one:prio}}
\label{app:prf-thmpolyallbutoneprio}
    \thmpolyallbutoneprio*
    \begin{proof}
    Let $D=(N,A)$ denote the underlying digraph of~$M$ where $(i,j) \in A$ if and only if $L_{ij}>0$. For a set~$F$ of arcs in~$D$, let $\chi^F \in \{0,1\}^{N \times N}$ denote the characteristic vector of~$F$, i.e., $\chi^F_{ij}=1$ if and only if $(i,j) \in F$ and $\chi^F_{ij}=0$ otherwise.
        
    We start with guessing the single bank~$b$ for which some compression~$C$ yields a clearing vector~$p$ for~$M-C$ with $\dft^{M-C}(p) \subseteq \{b\}$.     
    We also guess the lowest priority group $g_b^{\gamma}$ whose liabilities are fully paid by~$b$ under~$p$ in~$M-C$.
    We set $\widetilde{N}=N \setminus \{b\}$.
    We partition the set of agents in~$\widetilde{N}$ as 
    \[N^+=\bigcup_{i=1}^{\gamma} g_b^i, \qquad 
    N^\propto  = g_b^{\gamma+1}, 
    \qquad 
    N^0=\bigcup_{i=\gamma+2}^{k_i} g_b^i. 
    \]
    We call a compression~$C$ \emph{suitable} if $M-C$ admits a priority-proportional clearing vector~$p$ with $\dft^{M-C}(p) \subseteq \{b\}$ and where $N^+$ is exactly the set of banks whose liabilities from~$b$  are fully paid under~$p$ in~$M-C$.
    
    Due to Observation~\ref{obs:remove-cycles}, we may assume w.l.o.g.\ that the graph~$D$ underlying our financial market $M=(N,L,e,\alpha,\beta)$ is such that $G[\widetilde{N}]$ is acyclic: 
    Otherwise, we can repeat the following step as long as possible: if $C$ is a cycle containing only banks in~$\widetilde{N}$, then we replace~$M$ with $M-C$. By Observation~\ref{obs:remove-cycles}, there is a clearing vector~$p$ for~$M$ under which no bank in~$\widetilde{N}$ defaults if and only if there is a clearing vector for~$M-C$ under which no bank in~$\widetilde{N}$ defaults.
    Hence, $M$ admits a suitable compression if and only if $M-C$ admits one.

    Consequently, there is a compression where all banks remain solvent if and only if all banks remain solvent under a clearing vector for the original market~$M$. As we can efficiently decide whether this is the case, we further assume that $b$ defaults after any suitable compression.

    \smallskip
    We proceed by computing for each $i \in \widetilde{N}$ the value $q_i=\max\{0,\sum_{j \in N} L_{ij}-e(i)-\sum_{j \in \widetilde{N}}L_{ji}\}$ which can be thought of as the payment that $i$ needs to receive from~$b$ in order to become solvent (in the market without any compression). Note that if $q_i > L_{bi}$, then $i$ defaults even if $b$ can pay all its liabilities towards~$i$, a contradiction to our assumption on~$b$. Hence, we may assume that $q_i\leq L_{bi}$ for each~$i \in \widetilde{N}$.
    
    We construct a flow network~$D'$ from~$D$ as follows: we first 
    set the capacity of each arc~$(i,j)$ as~$L_{ij}$. Then we split $b$ into~$b^-$ and~$b^+$ in the usual way, i.e., by replacing each arc $(b,i)$ with $(b^+,i)$ and each arc~$(j,b)$ with~$(j,b^-)$. We set~$b^+$ as the source and $b^-$ as the sink of the network. Notice that each compression for~$M$ corresponds to a flow in~$D'$ and vice versa; we write $M^f$ for the financial market obtained by applying the compression determined by some flow~$f$ in~$D'$. We say that some flow in~$D'$ is \emph{suitable} if and only if the corresponding compression is suitable.
    
    Let $f$ be a flow in~$D'$, let $|f|$ denote its \emph{size} (i.e., the total flow value leaving~$b^+$),  and let $p$ be a priority-proportional clearing vector for~$M^f$. 
    Furthermore, let 
    \[|f^+|=\sum_{i \in N^+} f_{bi}, \qquad 
    |f^\propto|  = \sum_{i \in N^\propto} f_{bi},
    \qquad 
    |f^0|=\sum_{i \in N^0} f_{bi}. 
    \]    
    The sum of liabilities in~$M^f$ entering~$b$ is  $\sum_{i \in \widetilde{N}}L_{ib}-|f|$,
    and since all banks other than~$b$ remain solvent, this equals the total incoming payment of~$b$ in under~$p$. Thus, the total payment leaving~$b$ under~$p$ is
    \[
    p_b=\alpha_b e(b)+\beta_b\bigg(\sum_{i \in \widetilde{N}}L_{ib}-|f|\bigg).
    \]  
    Since $b$ fully pays its liabilities towards all banks in~$N^+$, which amounts to $\sum_{i \in N^+} (L_{bi} -f_{bi})$, 
    the total payment leaving~$b$ towards banks in~$N^\propto$ is exactly
    \[
    \alpha_b e(b)+\beta_b\bigg(\sum_{i \in \widetilde{N}}L_{ib}-|f|\bigg)-
    \sum_{i \in N^+} (L_{bi} -f_{bi})
    =A_b-\beta_b |f|+|f^+|
    \]  
    where we let $A_b=\alpha_b e(b)+\beta_b (\sum_{i \in \widetilde{N}}L_{ib}) - \sum_{i \in N^+} L_{bi}$. 
    Note that if we guessed $g_i^\gamma$ correctly, then we have that 
    \begin{equation}
    \label{eq:prop-money-bounds}
        0 \leq A_b-\beta_b |f|+|f^+| < L_b^\propto -|f^\propto|
    \end{equation}
    where $L_b^{\propto}=\sum_{i \in N^\propto} L_{bi}$.
     
    \myparagraph{Conditions for solvency.}
    Next, let us formulate the necessary and sufficient conditions for each $i \in \widetilde{N}$ to be solvent under the priority-proportional clearing vector~$p$ for $M^f$.
    \begin{itemize}
    \item Banks in~$N^+$: 
    Recall that we assumed $q_i \leq L_{bi}$ for each $i \in \widetilde{N}$. Since $b$ pays its full liability $L_{bi}$ towards each bank $i \in N^+$, we know that these banks remain solvent. Keep in mind, however, the condition
    \begin{equation}
    \label{eq:priority-banks-ok}
        A_b-\beta_b |f|+|f^+| \geq 0
    \end{equation}
    expressing the requirement that indeed all banks in~$N^+$ can be full paid by~$b$.
    \item Banks in~$N^\propto$: 
    For each bank~$i \in N^\propto$, bank~$i$ remains solvent under~$p$ in~$M^f$ if and only if
        \begin{equation}
        \label{eq:prop-banks-ok}
    \frac{L_{bi}-f_{bi}}{L_b^\propto -|f^\propto|}
    \cdot (A_b-\beta_b |f|+|f^+|) \geq q_i -f_{bi}. 
        \end{equation}
    \item Banks in~$N^0$: Since the payment from~$b$ towards these banks is zero, we obtain that $i \in N^0$ remains solvent if and only if 
    \begin{equation}
    \label{eq:lowprio-banks-ok}
        f_{bi} \geq q_i.
    \end{equation} 
    \end{itemize}

    \begin{claim}
    \label{clm:suitable-flow}
        A suitable flow~$f$ satisfies~(\ref{eq:priority-banks-ok}), (\ref{eq:prop-banks-ok}) for each $i \in N^\propto$, and (\ref{eq:lowprio-banks-ok}) for each $i \in N^0$.
        Conversely, a flow satisfying these conditions yields a flow for which $M^f$ admits a priority-proportional clearing vector.
    \end{claim}
    \begin{claimproof}
        The first statement follows from the discussion above. For the second statement, observe that the inequalities guarantee that each priority group receives enough payment to ensure that all banks remain solvent. Note, however that it is possible that the payments even suffices to pay all banks in~$N^\propto$ fully, in which case the $f$ is not suitable according to our definitions (which requires the last priority group to be fully paid by~$b$ to be~$g_b^{\gamma}$), but $M-C$ still admits a priority-proportional clearing vector where only~$b$ defaults.
    \end{claimproof}
        
    \begin{claim}
    \label{clm:lowprio-banks-sharp}
        If there is a flow~$f$ that satisfies~(\ref{eq:priority-banks-ok}), (\ref{eq:prop-banks-ok}) for each $i \in N^\propto$, and (\ref{eq:lowprio-banks-ok}) for each $i \in N^0$, then there exists such a flow~$f$ that, in addition, satisfies 
        \begin{equation}
            \label{eq:lowprio-banks-equality}
            f_{bi}=q_i
        \end{equation}for each $i \in N^0$, and hence, also $|f^-|=q^0:=\sum_{i \in N^0} q_i$.
    \end{claim}

    \begin{claimproof}
        Observe that decreasing the value of~$f_{bj}$ for some $j \in N^0$ increases the left-hand side of~(\ref{eq:priority-banks-ok}), as well that of~(\ref{eq:prop-banks-ok}) for each $i \in N^\propto$; hence, such a decrease does not violate any of the required inequalities as long as (\ref{eq:lowprio-banks-ok}) remains true for~$j$, proving the claim.
    \end{claimproof}
    \smallskip

    \myparagraph{Valid flows.}
    By Claim~\ref{clm:lowprio-banks-sharp}, we can modify our network~$D'$ by setting~$q_i$ for each~$i \in N^0$ as both the lower bound and the capacity of the arc~$(b^+,i)$; let $D'$ be the obtained network.\footnote{Standard flow techniques allow us to use networks that where we can set a lower bound on the flow value on each arc in addition to the capacity value;}
    Let us say that a flow~$f$ in~$D'$ is \emph{valid} if it satisfies~(\ref{eq:priority-banks-ok}), (\ref{eq:prop-banks-ok}) for each $i \in N^\propto$, and (\ref{eq:lowprio-banks-equality}) for each $i \in N^0$. 
    In the rest of the proof, we show how to find a valid flow in~$D'$ assuming there exists one.

    \myparagraph{Case $\beta_b=1$.}
    If $\beta_b=1$, then constraint~(\ref{eq:prop-banks-ok}) can be re-formulated as follows:
    \[
    L_{bi}(A_b-q^0-|f^\propto|)-f_{bi}(A_b-q^0) \geq 
    (q_i-f_{bi})L_b^\propto -q^0|f^\propto|.
    \] 
    These constraints for $i \in N^\propto$, as well as those in~(\ref{eq:priority-banks-ok}) and~(\ref{eq:lowprio-banks-equality}) for each $i \in N^0$, depend linearly on the variables $\{f_{bi}\}_{i \in \widetilde{N}}$. Hence, adding these constraints to the standard linear program  describing a flow in~$D'$ and solving the obtained LP, we can find in polynomial time a flow~$f$ that satisfies all required constraints in Claim~\ref{clm:suitable-flow} for each $i \in \widetilde{N}$ whenever such a flow exists.
    Hence, from now on we assume $\beta_b<1$.

    \smallskip
    \myparagraph{Case $\beta_b<1$.}
    Let us define 
    \[\lambda_{f}=\frac{A_b-\beta_b |f|+|f^+|}{L_b^\propto-|f^\propto|}
    .\]
    Condition (\ref{eq:prop-banks-ok}) can then be rephrased as
    \begin{equation}
    \label{eq:lower-bound-as-convex_combination}
    \phi_i(f):=(1-\lambda_{f}) \cdot f_{bi} + \lambda_{f} \cdot L_{bi}
    \geq q_i .        
    \end{equation}

    \begin{claim}
    \label{clm:easy-valid-flow}
        Suppose $\beta_b<1$. 
        If there exists a flow~$f$ in~$D'$  that satisfies (\ref{eq:lowprio-banks-equality}) for each $i \in N^0$
        and has size 
        \[
        |f| \geq \frac{L_b^\propto -A_b+q^0}{1-\beta_b},
        \]            
        then such a flow is valid.
    \end{claim}
    \begin{claimproof}
    It suffices to observe that the bound on~$|f|$ 
    is equivalent to $A_b-\beta_b|f|+|f^+|\geq L_b^\propto-|f^\propto|$ 
    which, in turn, implies not only (\ref{eq:priority-banks-ok})
    but also~(\ref{eq:prop-banks-ok}) for each $i \in N^\propto$, due to our assumption $q_i \leq L_{bi}$.
    \end{claimproof}
    \smallskip

    Using Claim~\ref{clm:easy-valid-flow}, we proceed by searching for a flow that fulfills the conditions of the claim; if we find one such flow~$f$, then its validity guarantees a priority-proportional clearing vector for~$M^f$.
    Hence, from now on we may assume that every flow that satisfies~(\ref{eq:lowprio-banks-equality}) for each $i \in N^0$ have size less than 
    $\frac{L_b^\propto -A_b+q^0}{1-\beta_b}$. Thus, we can assume that 
    \begin{equation}s
        \label{eq:upper-money-bound}
        A_b-\beta_b|f|+|f^+| < L_b^\propto - |f^\propto|.
    \end{equation}
    Note that this also implies $0 \leq \lambda_{f} <1$. 


\begin{claim}
\label{clm:flowsize-prio}
    Suppose $\beta_b<1$.
    If there is a valid flow~$f$ in~$D'$,
    then there exists a valid flow~$f'$ in~$D'$ that, in addition, either is a maximum-size flow in~$D'$, or satisfies 
    \begin{equation}
    \label{eq:flowsize-notmax}
    |f'|=\frac{q^\propto+\beta_b q^0 -A_b}{1-\beta_b}+q^0    
    \end{equation}
    where $q^\propto=\sum_{i \in N^\propto} q_i$.
    %
\end{claim}

\begin{claimproof} 
    First, note that increasing 
    the value of~$f_{bj}$ for some $j \in N^+$ increases the left-hand side of~(\ref{eq:priority-banks-ok}), as well that of~(\ref{eq:prop-banks-ok}) for each $i \in N^\propto$; hence, such an increase does not violate any of the required inequalities. Hence, we know that we can always choose a flow~$f$ that satisfies the required conditions and, additionally, cannot be increased via any augmenting path that starts with any of the arcs leading from~$b^+$ to~$N^+$.

    To consider the possibility of increasing the flow value via augmenting paths starting with an arc~$(b^+,i)$ for some~$i \in N^\propto$, let us investigate the effects of an increase in~$f_{bi}$ for some $i \in N^\propto$.
    Note that 

     \[\pdv{\lambda_{f}}{f_{bi}}=\frac{A_b-\beta_b (L^\propto_b +|f|-|f^\propto|)+|f^+|}{(L_b^\propto-|f^\propto|)^2} = 
     \frac{A_b-\beta_b L^\propto_b +|f^+|-\beta_b(|f^+|+|f^0|)}{(L_b^\propto-|f^\propto|)^2}
    .\]
    
    We distinguish between two cases.
    
    {{\bf Case A:}} $A_b-\beta_b L^\propto_b +|f^+|-\beta_b(|f^+|+|f^0|) \geq 0$. In this case, $\pdv{\lambda_{f}}{f_{bi}}\geq 0$, meaning that $\lambda_{f}$ weakly increases if $f_{bi}$ increases.
    
    Assume that $f$ is not a maximum flow but maximizes $|f|$ among all valid flows. 
    Then, $f$ can be increased 
    Let $f'$ be a flow in~$D'$ obtained from~$f$ via increasing its value via an augmenting path that starts with the arc~$(b^+,i)$ with an infinitesimal value for some $i \in N^\propto$.
    We are going to show that $f'$ is valid, contradicting our choice of~$f$. This contradiction proves that $f$ must be a maximal flow.
    
    Observe that $f_{bj} \leq L_{bj}$ by our construction of the network~$D'$, 
    and 
    $\phi_j(f)$  as defined by (\ref{eq:lower-bound-as-convex_combination})
    is the convex combination of $f_{bj}$ and~$L_{bj}$ with coefficients $1-\lambda_{f}$ and $\lambda_{f}$, respectively. Hence, taking~$f'$ instead of~$f$, 
    not only does the coefficient for the larger value ($L_{bj}$) weakly increase (since $\lambda_{f'} \geq \lambda_{f}$), 
    but additionally, the smaller value ($f_{bj}$) also weakly increases (since $f'_{bj}\geq f_{bj}$ holds for each $j \in \widetilde{N}$). Thus, we get that 
    \[(1-\lambda_{f'}) \cdot f'_{bj} + \lambda_{f'} \cdot L_{bi}
    \geq q_j \]
    holds for each $j \in N^\propto$, which means that $f'$ satisfies~(\ref{eq:lower-bound-as-convex_combination}) for each $j \in N^\propto$, as required.

    Note also that since $\lambda_{f'}\ge\lambda_f\geq 0$, in particular, by $L_b^\propto -|f^\propto| \geq 0$ we know that $A_b-\beta_b|f|+|f^+|\geq 0$, that is, constraint~(\ref{eq:priority-banks-ok}) remains true as well. 
    Constraints~(\ref{eq:lowprio-banks-equality}) remain unaffected.
    This proves that $f'$ is valid.

    {\bf Case B: $A_b-\beta_b L_b^\propto + |f^+| -\beta_b(|f^+|+|f^0|< 0$. }
     In this case, $\odv{\lambda_{f}}{f_{bi}}< 0$, meaning that $\lambda_{f}$  increases if $f_{bi}$ decreases for some $i \in N^\propto$. 
     Using that $|f|=|f^+|+|f^\propto|+|f^0|$ we obtain 
     \[\lambda_{f}=
     \frac{A_b-\beta_b |f|+|f^+|}{L_b^\propto-|f^\propto|}=
     \beta_b+ \frac{A_b-\beta_b L_b^\propto +|f^+|-\beta_b(|f^+|+|f^0|)}{L_b^\propto-|f^\propto|}
     =\beta_b + \pdv{\lambda_f}{f_{bi}} (L_b^\propto-|f^\propto|)
    .\]
     Let us now examine how $\phi_i(f)$ changes if we decrease the flow value~$f_{bi}$ on some arc~$(b^+,i)$:
     \begin{align*}
     \pdv{\phi_i(f)}{f_{bi}} 
     &= \pdv{\left((1-\lambda_{f})  \cdot f_{bi} + \lambda_{f} \cdot L_{bi}\right)}{f_{bi}} =
     -\pdv{\lambda_f}{f_{bi}} \cdot {f_{bi}} + 
     (1-\lambda_{f}) +  \pdv{\lambda_{f}}{f_{bi}} \cdot L_{bi} \\
     &= 1-\beta_b- \pdv{\lambda_f}{f_{bi}} (L_b^\propto-|f^\propto|) + (L_{bi}-f_{bi}) \pdv{\lambda_f}{f_{bi}}
      \\
     &= 1 - \beta_b -\pdv{\lambda_f}{f_{bi}}
     \left( \sum_{j \in  N^\propto \setminus \{ i\}} (L_{bj} -f_{bj}) \right) 
     \geq 1-\beta_b > 0
     \end{align*}
     where we used that  $L_{bj} -f_{bj} \geq 0$ for each $j \in \widetilde{N}$,
     and our assumption that $\pdv{\lambda_f}{f_{bi}}
     < 0$.
     %
    Therefore,
    decreasing the flow value~$f_{bi}$ but leaving $f_{bj}$ unchanged for all $j \in \widetilde{N} \setminus \{i\} $ strictly decreases $\phi_i(f)$ but, due to the increase in~$\lambda_{f}$, increases $\phi_j(f)$ for each $j \in N^\propto     \setminus \{i\}$ . 

    Let $f'$ be a valid flow in~$D'$ 
    that 
    has minimum total size. 
   We claim that $\phi_i(f')=q_i$ for each~$i \in N^\propto$.
   Otherwise, $\phi_i(f')>q_i$ for some~$i \in N^\propto$, so by the above paragraph we can decrease $f'_{bi}$ slightly while maintaining the inequalities~(\ref{eq:lower-bound-as-convex_combination}) for each $j \in \widetilde{N}$. Note that this operation increases the left-hand side of (\ref{eq:priority-banks-ok}) and has no effect on~(\ref{eq:lowprio-banks-equality}). Thus, $f'$ is valid.
   However, this contradicts our choice of~$f'$, proving that 
   $f$ must fulfill (\ref{eq:lower-bound-as-convex_combination}) for each $i \in \widetilde{N}$ with equality, which also means that (\ref{eq:prop-banks-ok}) is satisfied with equality for each $i \in N^\propto$.
    Summing up 
    these equalities, we get the equation 
    \begin{align*}
    A_b -\beta_b |f|+|f^+| &= \left(\sum_{i \in N^\propto} q_i \right)-|f^\propto|;
    \\
    (1-\beta_b)(|f^\propto|+|f^+|) 
    &= q^\propto  
    +
    \beta_b q^0 - A_b
    \end{align*}
    Due to $\beta_b<1$, this leads to $|f^\propto|+|f^+|=\frac{q^\propto+\beta_b q^0 -A_b}{1-\beta_b}$, which proves (\ref{eq:flowsize-notmax}).
\end{claimproof}    
\smallskip

    Due to Claim~\ref{clm:flowsize-prio}, we may assume that we know~$|f|$. Then (\ref{eq:priority-banks-ok}) can be reformulated as a simple lower bound 
    \begin{equation}
    \label{eq:fplus-lower-bound}
        |f^+| \geq \beta_b |f|-A_b.
    \end{equation}
    Moreover, 
    (\ref{eq:prop-banks-ok}) 
    can be re-written as 
    \begin{equation}
    \label{eq:lower-bounds-simpler-prio}
        L_{bi}(A_b -\beta_b |f|+|f^+|)-f_{bi}(A_b -\beta_b|f|) \geq q_i (L_b^\propto-|f^\propto|)
        -f_{bi}L_b^\propto + f_{bi}(|f^+|+|f^\propto|)
        .
    \end{equation} 
    Note that (\ref{eq:lower-bounds-simpler-prio}) is a linear inequality over variables $f_{bi}$, $i \in N^+ \cup N^\propto$, because both $|f|$ and $|f^+|+|f^\propto|=|f|-q^0$ are known and thus can be considered constant. 
    Hence, by adding the constraints~(\ref{eq:fplus-lower-bound}), constraints (\ref{eq:lower-bounds-simpler-prio}) for each $i \in N^\propto$, and (\ref{eq:lowprio-banks-equality}) for each $i \in N^0$ to the standard linear program describing a flow in~$D'$ of size~$|f|$, 
    we can decide in polynomial time whether there exists a valid flow~$f$. If such a flow~$f$ exists, then each bank in~$\widetilde{N}$ is solvent under the proportional clearing vector for $M^f$. If no such flow exists, then we conclude that there is no suitable compression for~$M$. 
\end{proof}

\section{Additional simulations}
\label{app:more-sim}
In this section we present the results of simulations on synthetic data where liabilities and  endowments 
are not generated using uniform distributions but 
according to some log-normal distribution.
All other settings used for these simulations are as described in Section~\ref{sec:sim}.

Figure~\ref{fig:uni-liab} presents the simulation results where liabilities are generated using a uniform distribution from $[100,1000]$.
Figure~\ref{fig:uni-liab-uni-end} depicts the 
case when the endowment of each bank~$i$ chosen from a uniform distribution from $[0,0.8L_i]$ where $L_i$ is the total outward liability of bank~$i$; this
is exactly the setting described in Section~\ref{sec:sim}, repeated only to make comparisons easier.
By contrast, Figure~\ref{fig:uni-liab-uni-end} shows  results for endowments generated using a log-normal distribution to obtain a model with numerous small endowments and few much larger ones.
More precisely, the initial endowment of bank~$i$ is set to $0.8 L_i \cdot X$ where $X$ is a random noise  variable drawn from the log-normal distribution $X \sim \log(N(0,0.5))$.

It can be seen that under log-normal endowments (but uniform liabilities), the number of defaulting banks starts to decrease after a certain market size, and the effects of   compressions are also more significant, both for the greedy and the MILP method. We attribute this effect to the mismatch between the two distributions, which induces less defaulting banks, due to the relatively few banks with really large negative balance.
\begin{figure}[htbp]
    \centering
    \begin{subfigure}{0.49\textwidth}
\includegraphics[width=\textwidth]{pictures/scaling_defaults_synthetic.png}
    \subcaption{Results for uniform liabilities and uniform endowments.}
    \label{fig:uni-liab-uni-end}    
    \end{subfigure}
    \hfill
    \begin{subfigure}{0.49\textwidth}
    \includegraphics[width=\textwidth]{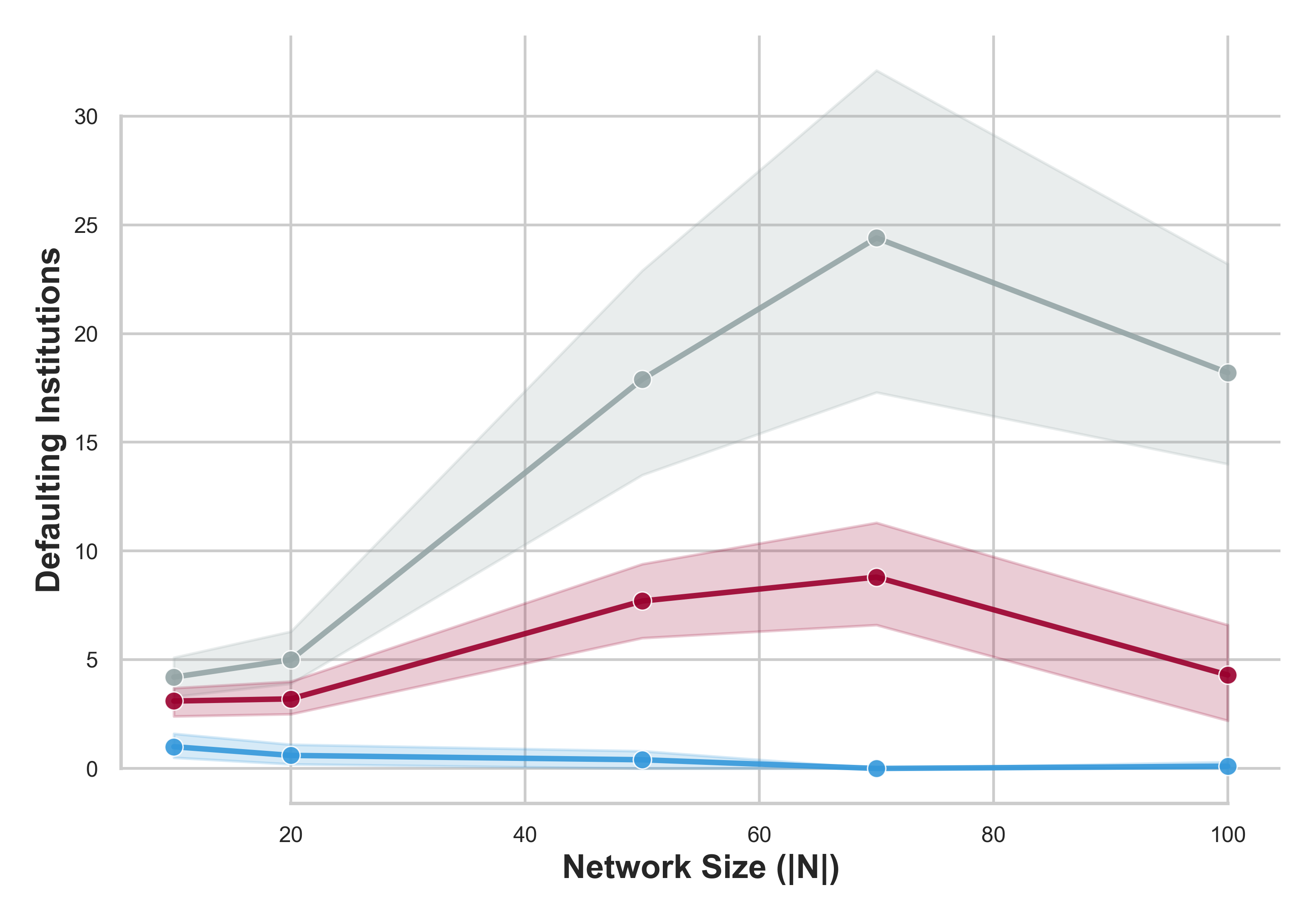}
    \subcaption{Results for uniform liabilities and log-normal endowments.}
    \label{fig:uni-liab-lognorm-end}    
    \end{subfigure}
    \caption{Simulation results for synthetic data generated with uniform liabilities.}
    \label{fig:uni-liab}
\end{figure} 

Next, Figure~\ref{fig:uni-lognorm} presents the simulation results where liabilities are generated using a log-normal distribution.
More precisely, liabilities are drawn randomly from the log-normal distribution $Y=\log(N(\mu,1))$ for some constant~$\mu$ such that $Y$ has mean~$200$.
Figure~\ref{fig:lognorm-liab-uni-end} depicts the 
case when endowments are chosen  uniformly from $[0,0.8L_i]$, whereas
Figure~\ref{fig:lognorm-liab-lognorm-end} shows our results for endowments generated by multiplying $0.8L_i$ with a log-normal noise $X \sim \log(N(0,0.5))$, as described for Figure~\ref{fig:uni-liab-lognorm-end}. 
We observe that for log-normal liabilities combined with uniform endowments, the number of defaulting banks without compression increases faster with the market size when compared to the case when liabilities are uniform as well (Figure~\ref{fig:uni-liab-uni-end}), while the greedy and MILP methods yield similar results, both performing slightly better under uniform liabilities. 
In the case when both liabilities and endowments are generated using log-normal distributions, we see results that are similar to the case when 
both these values are obtained from a uniform distribution (see again Figure~\ref{fig:uni-liab-uni-end}).
\begin{figure}[htbp]
    \centering
    \begin{subfigure}{0.49\textwidth}
\includegraphics[width=\textwidth]{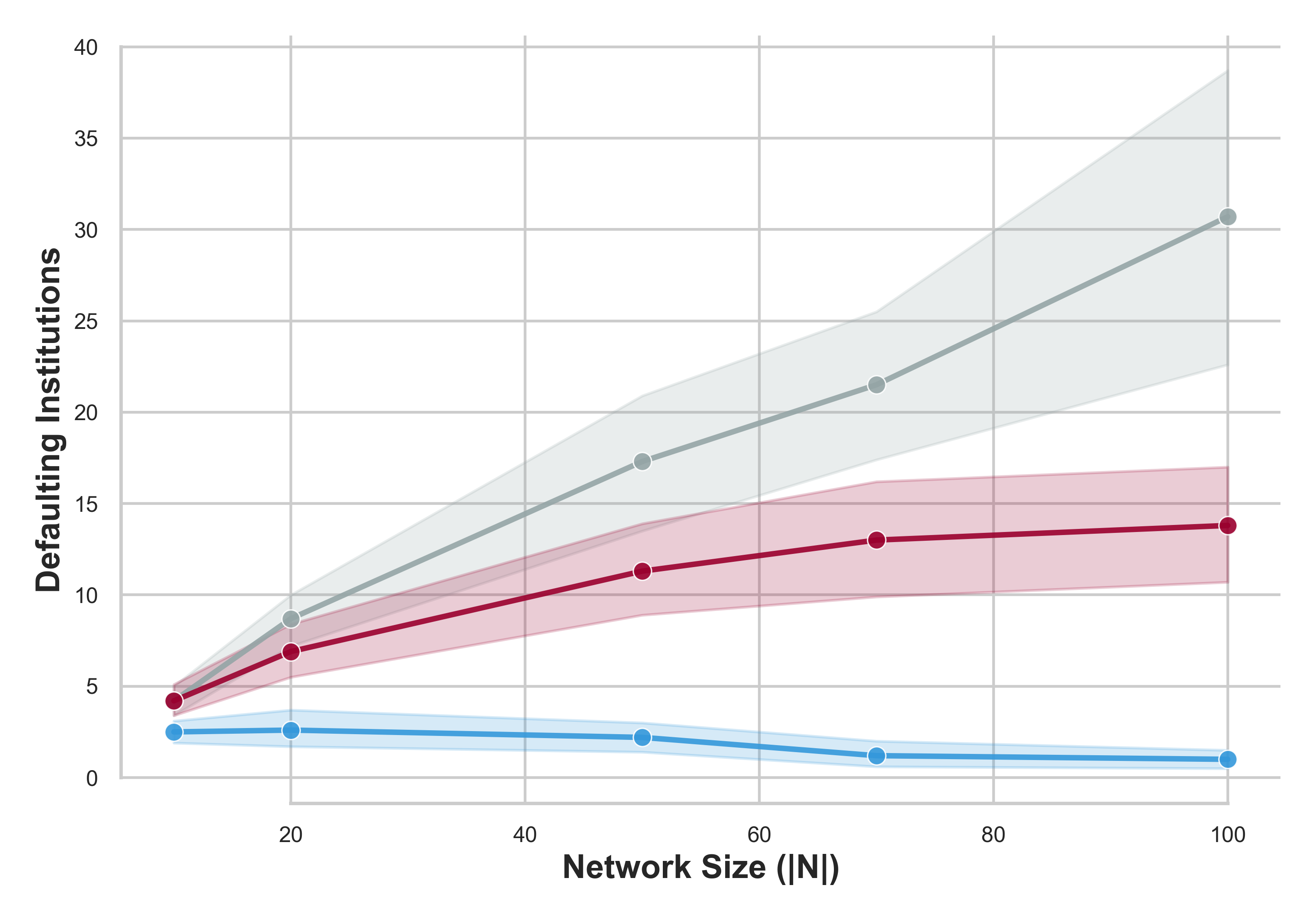}
    \subcaption{Results for log-normal liabilities and uniform endowments.}
    \label{fig:lognorm-liab-uni-end}    
    \end{subfigure}
    \hfill
    \begin{subfigure}{0.49\textwidth}
    \includegraphics[width=\textwidth]{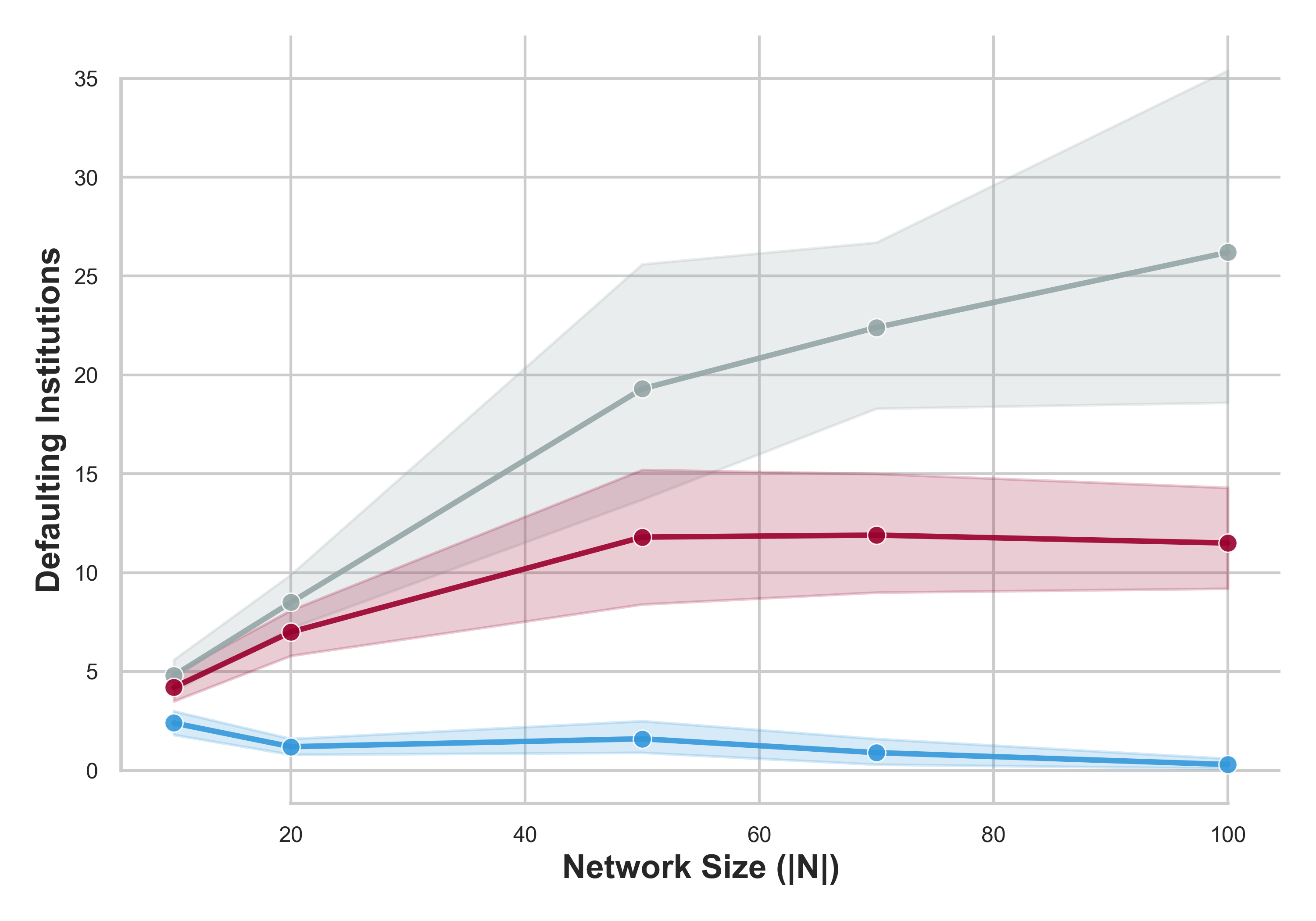}
    \subcaption{Results for log-normal liabilities and log-normal endowments.}
    \label{fig:lognorm-liab-lognorm-end}    
    \end{subfigure}
    \caption{Simulation results for synthetic data generated with log-normal liabilities.}
    \label{fig:uni-lognorm}
\end{figure}

\end{appendices}

\end{document}